\documentclass[letterpaper, 10 pt, conference]{ieeeconf}
\IEEEoverridecommandlockouts
\usepackage[utf8]{inputenc}
\usepackage{amsmath,amssymb,amsfonts,mathtools}

\usepackage{amsthm}
\usepackage{xcolor}
\usepackage{graphicx}
\usepackage{float}
\usepackage{babel}
\usepackage{algorithm}
\usepackage{algorithmic}
\usepackage{bbm}
\usepackage{booktabs}
\usepackage{caption}
\usepackage{subcaption}
\usepackage{array}
\usepackage[hidelinks]{hyperref}
\usepackage{multirow}
\usepackage{bm}

\newtheorem{theorem}{Theorem}
\newtheorem{lemma}{Lemma}
\newtheorem{proposition}{Proposition}
\newtheorem{corollary}{Corollary}

\newtheorem{assumption}{Assumption}
\newtheorem{remark}{Remark}
\newtheorem{designreq}{Design Requirement}

\newcommand{\R}{\mathbb{R}}
\newcommand{\E}{\mathbb{E}}
\newcommand{\Sigw}{\Sigma_w}
\newcommand{\Sigs}{\Sigma_s^*}
\newcommand{\norm}[1]{\left\|#1\right\|}
\newcommand{\normF}[1]{\left\|#1\right\|_F}
\newcommand{\tr}{\operatorname{tr}}
\newcommand{\Fcal}{\mathcal{F}}
\newcommand{\Ucal}{\mathcal{U}}
\newcommand{\Xcal}{\mathcal{X}}
\newcommand{\Ncal}{\mathcal{N}}
\newcommand{\hSigk}[1]{{\hat{\Sigma}_{#1}}^{(k)}}
\newcommand{\hSigkp}[1]{{\hat{\Sigma}_{#1}}^{(k+1)}}
\newcommand{\Lcal}{\mathcal{L}}

\title{Adaptive MPPI with Online Disturbance Covariance Estimation:
Provable Stability Tightening via Spatial Smoothing}

\author{
Hyung-Jin Yoon$^{\dagger}$ and
Hunmin Kim$^{\ddagger}$
\thanks{$^{\dagger}$H.-J. Yoon is with the
Department of Mechanical and Nuclear Engineering, Tennessee Technological
University, Cookeville, TN, USA.}%
\thanks{$^{\ddagger}$H. Kim is with the School of Engineering, Department
of Electrical and Computer Engineering, Mercer University, Macon, GA, USA.}%
\thanks{This work was supported by internal funding at Tennessee
Technological University.}
\thanks{\thanks{This paper is part of a companion series on the closed-loop
stability of Model Predictive Path Integral (MPPI) control. Manuscript
submitted June 2026.}}
}

\begin{document}
\maketitle

% ------------------------------------------------------------------
\begin{abstract}
We address Model Predictive Path Integral (MPPI) control for systems
subject to additive process disturbances whose covariance is unknown
and spatially varying. An incorrect disturbance covariance induces a
persistent mismatch penalty in the closed-loop stability certificate,
whereas online covariance estimation can reduce this penalty over
time. We propose an online cell-wise disturbance-covariance estimator
with spatial diffusion and prove that it converges to a smoothed fixed
point, with an error bound separating three effects: a
stochastic-approximation term that decreases as more closed-loop data
are collected, a spatial-smoothing bias caused by diffusion across
neighboring cells, and a drift term caused by slow temporal variation
of the disturbance field. Choosing the diffusion kernel proportional
to the stationary visitation measure makes the spatial diffusion term
dissipative in the weighted norm used in the Lyapunov analysis.
Substituting the resulting covariance estimate into the MPPI sampling
distribution yields a time-varying adaptation penalty that quantifies
the transient cost of learning in the closed-loop bound; this penalty
decreases with the stochastic-approximation component and, over any
prescribed finite horizon, is bounded by the sum of a residual
smoothing bias and an accumulated drift allowance. The main result is
a \emph{payoff theorem}: although the adaptive controller may
initially pay a larger transient penalty while learning, for any
fixed mismatched disturbance covariance whose mismatch exceeds this
residual smoothing-bias-plus-drift-allowance bound, there exists a
finite, computable crossover time, within the prescribed horizon,
after which the adaptive stability certificate is strictly tighter.
In contrast to heuristic covariance adaptation, the proposed estimator
treats the disturbance covariance as a statistical estimand and
updates the MPPI sampling covariance and corresponding control penalty
together, preserving the path-integral structure throughout
adaptation.
\end{abstract}

% -------------------------------------------------------
\section{Introduction}
\label{sec:intro}
% -------------------------------------------------------
Model Predictive Path Integral (MPPI) control~\cite{williams} solves
finite-horizon stochastic optimal control by drawing $M$ parallel
sample trajectories and returning an importance-weighted control
update. Because it requires no gradient of the dynamics or cost,
MPPI has been applied to off-road navigation~\cite{williams}, legged
locomotion, and aerial vehicles.

Companion work has begun to develop formal closed-loop stability
guarantees for MPPI. For linear time-invariant systems, one companion
paper~\cite{p1lti} establishes exponential stability in expectation up
to residual floors caused by process noise, finite-sample MPPI
approximation error, and sampling-confidence effects, via a
finite-sample MPPI--LQR approximation bound and a Lyapunov perturbation
argument. A second companion paper extends the analysis to nonlinear
systems using contraction theory and Control Lyapunov Functions (CLFs),
yielding a finite-horizon, high-probability localized mean practical
stability bound with the same qualitative structure~\cite{p2nonlinear}.

Both certified bounds share a common shape: an exponentially decaying
nominal term plus additive residual terms. Among these residuals, the
process-noise contribution depends explicitly on the covariance
$\Sigw$, whereas increasing the MPPI sample count mainly reduces the
Monte Carlo approximation error and does not remove the disturbance
floor. So once the finite-sample MPPI error is controlled, inaccurate
knowledge of $\Sigw$ becomes a direct limitation on the tightness of
the stability certificate. The present paper closes this gap by
introducing an online adaptive covariance estimator together with
finite-sample convergence guarantees and a plug-in stability analysis.

In practice $\Sigw$ is rarely known and rarely uniform: terrain
variation, wind gusts, and actuator wear produce spatially
heterogeneous, slowly time-varying noise. The consequences of a wrong
noise estimate are structural, not merely quantitative. Because MPPI
enforces the coupling constraint
$\sigma\sigma^\top = \lambda_0(R^\top R)^{-1}$~\cite{med}, the sampling
covariance and control cost weight are not independent: any mismatch
between the nominal $\hat\sigma$ and the true $\sigma$ propagates
through both simultaneously. As derived in~\cite{med} (eqs.~(12)--(13)),
overestimating $\hat\sigma > \sigma$ decreases the effective control
cost weight $\|(RR^\top)(\hat\sigma)\|$ and the risk sensitivity
$\alpha(\hat\sigma)$ at the same time, weakening the controller against
the true disturbance; underestimation reverses both effects, producing
over-aggressive control.

Our prior work~\cite{med} addressed this by modeling the selection of
$(\hat\sigma, \lambda_\alpha)$ as a Markov Decision Process and
training an actor-critic agent to choose these parameters online. The
MDP formulation was rigorous: it proved that the state transition is
Markov (Proposition~3 of~\cite{med}) and demonstrated performance
improvements over static risk-sensitive MPPI on double integrator and
bicycle model environments. However, the RL policy treats $\hat\sigma$
as a control parameter rather than a statistical estimand, so it
provides no convergence rate, no characterization of steady-state
error, and no connection to the closed-loop stability bound.

This paper provides that connection. Under slow time variation
(Assumption~\ref{ass:slowvar}), spatial Lipschitz continuity
(Assumption~\ref{ass:smooth}), and the stability and regularity
conditions stated below, we prove that a cell-wise recursive estimator
with stationary-measure-weighted spatial diffusion converges to a
smoothed covariance field with an explicit bias--variance--drift
decomposition, and we show how the resulting estimation error enters
the companion nonlinear MPPI stability certificate as an additive
adaptation penalty. This yields a precise payoff statement: after a
computable crossover time within any prescribed finite horizon, the
adaptive covariance estimate gives a strictly tighter certified
stability bound than any fixed covariance whose mismatch exceeds the
estimator's residual smoothing bias and finite-horizon drift
allowance. We make four formal contributions.

\begin{enumerate}
\item \textbf{Fixed-point characterization
(Lemma~\ref{lem:fixedpoint}).}
We characterize the unique fixed point of the expected estimator
dynamics under stationary-measure-weighted spatial diffusion. The
fixed point is a smoothed version of the true covariance field, with
cell-wise bias bounded by
$\max_s\normF{\Sigs - \Sigw(s)} \leq \beta d_\kappa L_\sigma
r_{\Ncal} / p_{\min}$.

\item \textbf{Estimator convergence
(Theorem~\ref{thm:convergence}).}
We prove that the online estimator converges toward this smoothed
fixed point with an explicit three-term error bound consisting of an
$O(1/\sqrt{k})$ stochastic-approximation term, a spatial-smoothing
bias, and a finite-horizon accumulated-drift term. The proof uses a
$p_s$-weighted Lyapunov function and the dissipativity of the
diffusion operator (Lemma~\ref{lem:dissipative}), avoiding
two-timescale stochastic approximation machinery.

\item \textbf{Plug-in stability analysis
(Proposition~\ref{prop:adaptive}).}
We show that substituting the online estimate
$\hSigk{s(x_k)}$ for the unknown process covariance in the companion
nonlinear MPPI stability bound introduces an additive adaptation
penalty $\psi^{(k)}$. The stochastic-approximation component of this
penalty decays as $O(1/\sqrt{k})$, while the remaining terms are
explicitly controlled by the spatial-smoothing bias and the
time-variation rate.

\item \textbf{Payoff theorem
(Theorem~\ref{thm:payoff}).}
We prove that, for any fixed horizon $T$ and any fixed mismatched
covariance $\bar\Sigma_w$ whose mismatch exceeds the estimator's
smoothing bias plus its accumulated drift allowance over $[0,T]$,
there exists a finite crossover time $k_T^*\le T$ after which the
adaptive controller achieves a strictly tighter certified stability
bound. The crossover time is given in closed form.
\end{enumerate}
% --------------------------------------------------------

% -------------------------------------------------------
\section{Related Work}
\label{sec:related}
% -------------------------------------------------------

\textbf{MPPI and path-integral control.}
MPPI builds on the path-integral formulation of stochastic optimal
control, in which the control update is computed from
importance-weighted sampled trajectories rather than from gradients
of the dynamics or cost~\cite{williams}. This structure has made MPPI
attractive for robotic systems with nonlinear dynamics, nonconvex
costs, and real-time parallel rollout requirements. Most existing
theoretical work on MPPI focuses on open-loop sampling accuracy,
optimizer convergence, or the performance of a single planning
update. The closed-loop stability question is more delicate because
finite-sample approximation errors are injected repeatedly through
receding-horizon execution.

\textbf{Covariance adaptation in path-integral control.}
Stulp and Sigaud~\cite{pi2} introduced PI$^2$, which adapts the
sampling covariance $\Sigma_\epsilon$ during policy improvement. In
standard path-integral control, however, the sampling covariance and
the control-cost weight are not independent: the path-integral
matching condition ties $\Sigma_\epsilon$ to the control weighting
matrix through the coupling constraint~\eqref{eq:sigR}, so adapting
$\Sigma_\epsilon$ while holding $R$ fixed implicitly changes the
effective control penalty and risk sensitivity at every update. Our
prior risk-sensitive MPPI work~\cite{med} exploited this coupling as
an adaptive design mechanism. In contrast, the present paper treats
the process covariance $\Sigw$ as a statistical estimand: the
estimator targets $\Sigw$ directly and then updates both
$\Sigma_\epsilon$ and $R$ consistently so that~\eqref{eq:sigR} is
maintained throughout adaptation.

\textbf{Noise covariance estimation in MPC.}
Noise covariance estimation has a long history in model predictive
control and filtering. Odelson et al.~\cite{odelson} proposed an
autocovariance least-squares method for estimating process and
measurement noise covariances in linear MPC settings. Such methods
are typically offline or batch procedures designed for fixed linear
models. The present setting differs in three respects: the covariance
is spatially varying, it may drift slowly over time, and the estimate
is used online inside a nonlinear sampling-based controller. These
differences require a recursive estimator whose error can be
propagated into a closed-loop MPPI stability certificate.

\textbf{Unknown noise in control sample complexity.}
Dean et al.~\cite{dean} studied the sample complexity of LQR under
unknown dynamics and noise statistics, showing how statistical
uncertainty in system quantities affects controller synthesis and
performance. Our objective is narrower but complementary: we do not
identify the full dynamics and assume the nominal model used by MPPI
is available, instead estimating the process-noise covariance field
online. In the single-cell, time-invariant limit, the
stochastic-approximation component of our estimator has the standard
$O(1/\sqrt{k})$ statistical rate, consistent with the usual
$O(\epsilon^{-2})$ sample complexity scaling for estimating second
moments.

\textbf{Stochastic approximation and spatial diffusion.}
Classical stochastic approximation provides convergence guarantees
for recursive estimators with diminishing step sizes. Coupled
stochastic approximation schemes are often analyzed using
two-timescale methods, as in Borkar~\cite{borkar}, where one
recursion evolves sufficiently slowly relative to another. Our
estimator instead combines a diminishing data-driven covariance
update with a fixed spatial diffusion term: choosing the diffusion
kernel from the stationary visitation measure makes the diffusion
operator dissipative in the $p_s$-weighted norm, so the stochastic
approximation and diffusion terms can be handled by a single
Lyapunov argument rather than a two-timescale separation.

\textbf{Closed-loop stability of MPPI.}
Companion work established the closed-loop stability foundation used
here. The linear companion paper~\cite{p1lti} proves exponential mean
stability for LTI systems up to residual floors caused by process
noise, finite-sample MPPI error, and sampling-confidence effects. The
nonlinear companion paper~\cite{p2nonlinear} extends the argument
using contraction theory and Control Lyapunov Functions, yielding a
finite-horizon, high-probability localized mean practical stability
bound. Both certificates contain a process-noise contribution that
depends explicitly on $\Sigw$. The present paper builds on this
interface: by estimating $\Sigw$ online, we convert covariance
mismatch into an explicit adaptation penalty and prove when this
penalty becomes smaller than the permanent mismatch penalty of any
fixed, non-adaptive covariance choice.

\subsection*{Paper Organization}

Section~\ref{sec:setup} defines the system, discretization, and
constrained MPPI law. Section~\ref{sec:mainresults} presents the
main results: the estimator and its fixed-point analysis
(Sec.~\ref{sec:estimator}), its convergence guarantee
(Sec.~\ref{sec:convergence}), the resulting adaptive stability bound
(Sec.~\ref{sec:plugin}), and the payoff theorem
(Sec.~\ref{sec:payoff}). Section~\ref{sec:simulation} validates on
two environments. Section~\ref{sec:discussion} discusses
limitations.

% ------------------------------------------------------------------
\section{Problem Setup and Notation}
\label{sec:setup}
% ------------------------------------------------------------------

We consider the nonlinear stochastic discrete-time system
\begin{equation}
x_{k+1} = f(x_k,u_k) + w_k,
\label{eq:dynamics}
\end{equation}
where $x_k\in\R^n$, $u_k\in\R^m$, and
$f:\R^n\times\R^m\to\R^n$ is continuously differentiable with
$f(0,0)=0$. The process disturbance is conditionally Gaussian with
state-dependent covariance:
\begin{equation*}
w_k \mid x_k=x \sim \mathcal{N}\left(0,\Sigw^{(k)}(x)\right),
% \label{eq:spatialnoise}
\end{equation*}
where $\Sigw^{(k)}:\R^n\to\mathbb{S}_{++}^n$ is an unknown,
spatially smooth, slowly time-varying positive-definite covariance
field. When the covariance is constant in both space and time, this
recovers the additive-noise setting considered in the companion
nonlinear stability analysis~\cite{p2nonlinear}. For notational
simplicity, we write $\Sigw(x)$ when the time index is clear from
context.

\subsection{State Space Discretization}

Let $\Xcal\subset\R^n$ be a compact operating region. Partition
$\Xcal$ into $S$ non-overlapping cells $\{C_s\}_{s=1}^S$ with
diameter
\begin{equation*}
  r := \max_{1\le s\le S}\operatorname{diam}(C_s).  
\end{equation*}
Let $s(x)$ denote the index of the cell containing $x$, and let
$c_s$ denote a representative point of $C_s$, such as its centroid.
Within each cell, we approximate the spatially varying covariance by
the cell representative
\begin{equation*}
  \Sigw^{(k)}(x) \approx {\Sigw^{(k)}}_s
  := \Sigw^{(k)}(c_s),
  \qquad x\in C_s .  
\end{equation*}
By the spatial Lipschitz condition in Assumption~\ref{ass:smooth},
\begin{equation*}
  \normF{\Sigw^{(k)}(x)-{\Sigw^{(k)}}_s}
  \le L_\sigma \operatorname{diam}(C_s)
  \le L_\sigma r,
  \qquad x\in C_s,  
\end{equation*}
so the cell-discretization error is $O(L_\sigma r)$ and vanishes as
$r\to0$.

\subsection{MPPI Control Law and the $\sigma$-$R$ Constraint}

At each step $k$, MPPI draws $M$ i.i.d.\ perturbation sequences with
$\epsilon_i^{(j)}\!\sim\!\mathcal{N}(0,\Sigma_\epsilon)$ and returns:
\begin{equation}
\begin{aligned}
  w^{(j)} &= \exp\!\bigl(-J(x_k,U^{(j)})/\lambda\bigr), \\
  u_k^{\mathrm{MPPI}} &=
  \frac{\sum_{j=1}^M w^{(j)} v_k^{(j)}}{\sum_{j=1}^M w^{(j)}}.  
\end{aligned}\label{eq:mppi}
\end{equation}
For the path integral HJB to admit a forward-sampling solution, the
coupling constraint must hold~\cite{med}:
\begin{equation}
  \Sigma_\epsilon = \lambda_0\,(R^\top R)^{-1}.
  \label{eq:sigR}
\end{equation}
When $\Sigw$ is unknown we set $\Sigma_\epsilon = \hSigk{s(x_k)}$
and update $R^{(k)} = \lambda_0^{1/2}(\hSigk{s(x_k)})^{-1/2}$
at each step, maintaining~\eqref{eq:sigR} throughout adaptation.
This is the key structural departure from PI$^2$-CMA~\cite{pi2}:
there, covariance adaptation is used to tune exploration in policy
improvement, whereas here the covariance is treated as an estimate of
the process disturbance and is coupled with the corresponding control
penalty through~\eqref{eq:sigR}.

\subsection{Notation}

Throughout the paper, $\norm{\cdot}$ denotes the Euclidean norm for
vectors, and $\normF{\cdot}$ denotes the Frobenius norm for matrices.
The set $\mathbb{S}_{++}^n$ denotes the cone of $n\times n$ symmetric
positive-definite matrices. The notation $\Ncal(s)$ denotes the
neighbor set of cell $s$, and $\Fcal_k$ denotes the sigma-algebra
generated by all information available up to time $k$.

\subsection{Assumptions}
% ------------------------------------------------------------------
\begin{assumption}[Slow Time Variation]
\label{ass:slowvar}
For all cells $s$ and all $k\geq0$,
\[
  \normF{\Sigw^{(k+1)}(s)-\Sigw^{(k)}(s)}
  \leq \epsilon_v,
\]
where $\epsilon_v\geq0$. When $\epsilon_v=0$, the noise field is
time-invariant.
\end{assumption}
% ------------------------------------------------------------------

% ------------------------------------------------------------------
\begin{assumption}[Spatial Smoothness]
\label{ass:smooth}
For any two cells $s,s'$ and all $k\geq0$,
\[
  \normF{\Sigw^{(k)}(s)-\Sigw^{(k)}(s')}
  \leq L_\sigma d(s,s'),
\]
where $d(s,s')$ is the Euclidean distance between the centroids of
$C_s$ and $C_{s'}$, and $L_\sigma\geq0$.
\end{assumption}
% ------------------------------------------------------------------

\begin{assumption}[Geometric Ergodicity]
\label{ass:geometric}
The closed-loop process generated by~\eqref{eq:dynamics} under MPPI,
restricted to the compact operating region $\Xcal$, is geometrically
ergodic over the cells $\{C_s\}_{s=1}^S$: there exist stationary
visitation probabilities $\{p_s\}_{s=1}^S$ with
\[
  p_{\min}:=\min_{1\le s\le S}p_s>0,
\]
and constants $C_{\mathrm{erg}}<\infty$ and $\rho\in(0,1)$ such that
\[
  \bigl|\Pr[x_k\in C_s\mid x_0=x] - p_s\bigr|
  \le C_{\mathrm{erg}}\,\rho^{k},
  \, \forall x\in\Xcal,\ \forall s,\ \forall k\ge0.
\]
This bound is taken to hold uniformly over the (slowly varying)
closed-loop policy induced by any admissible estimator value
$\hat\Sigma\in\Theta^S$: since $\hat\Sigma^{(k)}$ changes by only
$O(\alpha^{(k)})$ per step (Sec.~\ref{sec:estimator}), the cell
process is approximately time-homogeneous over any window of
length $\tau=O(\log k)$, and Lemma~\ref{lem:decorrelation} below
uses the assumption only over such short windows.
\end{assumption}

\begin{remark}[Role of the Ergodicity Assumption]
Assumption~\ref{ass:geometric} is a persistent-excitation condition,
as in classical adaptive control: a parameter can only be identified
along directions the closed-loop trajectory actually visits. Here
the ``direction'' is a cell of the operating region, so no estimator
can recover $\Sigw(s)$ in a cell with $p_s=0$. $p_{\min}>0$ ensures
every cell is visited infinitely often, while $\rho$ controls how
fast empirical visitation converges to $\{p_s\}$. This is the
standard regime for robotic tasks with recurrent coverage --
patrolling, lawnmower inspection, and repeated docking routes all
revisit a bounded region by design, as in the lawnmower trajectory
of Section~\ref{sec:simulation}. For non-recurrent tasks, $\Xcal$
can instead denote a corridor revisited across episodes, or an
initial calibration phase can seed the estimate before deployment.
\end{remark}

\begin{designreq}[Reversible Diffusion Kernel]
\label{req:kernel}
Let $\{p_s\}_{s=1}^S$ be the stationary visitation probabilities of
Assumption~\ref{ass:geometric}. The cell-neighbor graph is undirected:
if $s'\in\Ncal(s)$, then $s\in\Ncal(s')$. We require the diffusion
kernel $\kappa(s,s')$ to be nonnegative and to satisfy the
detailed-balance identity
\[
  p_s\kappa(s,s') = p_{s'}\kappa(s',s),
  \qquad s'\in\Ncal(s).
\]
Equivalently, we choose symmetric edge weights
$a_{ss'}=a_{s's}\ge0$ and set
\[
  \kappa(s,s') = \frac{a_{ss'}}{p_s},
  \qquad s'\in\Ncal(s).
\]
Moreover,
\[
  d_\kappa
  :=
  \max_{1\le s\le S}
  \sum_{s'\in\Ncal(s)} \kappa(s,s')
  <\infty .
\]
\end{designreq}

\begin{remark}[Design Nature of Requirement~\ref{req:kernel}]
Unlike Assumptions~\ref{ass:slowvar}--\ref{ass:geometric}, which
constrain the unknown, physical disturbance field, Requirement~\ref{req:kernel}
constrains only the diffusion kernel $\kappa$, which is chosen by the
algorithm designer and never touches the true process noise. The kernel
governs only how the \emph{estimator} shares information across cells;
it induces no smoothing of the environment itself, and the true field
$\Sigw^{(k)}(\cdot)$ is never altered by it. Because $\kappa(s,s')$ is
under our control, the detailed-balance identity is enforced by
construction: any symmetric edge weights $a_{ss'}=a_{s's}$ and the
choice $\kappa(s,s')=a_{ss'}/p_s$ satisfy it automatically, with no
constraint on $\Sigw^{(k)}(\cdot)$ required.
\end{remark}

% ------------------------------------------------------------------
\begin{assumption}[Observed Process Residuals: Unbiasedness and Bounded Fourth Moment]
\label{ass:noise}
At each time $k$, after applying $u_k$, the next state $x_{k+1}$ is
observed and the nominal model $f$ is available, so the realized
process residual
\[
  \hat w_k := x_{k+1}-f(x_k,u_k)
\]
can be computed. Conditional on $x_k\in C_s$ (equivalently $s_k=s$),
this residual satisfies:
\begin{enumerate}
  \item[(i)] \textbf{Conditional unbiasedness.}
  \[
    \E\big[\hat w_k\hat w_k^\top \mid x_k\in C_s\big] = \Sigw^{(k)}(s).
  \]
  \item[(ii)] \textbf{Bounded fourth moment.} There exists
  $\sigma_\xi^2<\infty$, independent of $k$ and $s$, such that
  \[
    \E\Big[\normF{\hat w_k\hat w_k^\top - \Sigw^{(k)}(s)}^2
    \,\Big|\, x_k\in C_s\Big]
    \le \sigma_\xi^2.
  \]
  This is a bound on a fourth moment of $\hat w_k$: since
  $\hat w_k\hat w_k^\top$ is already quadratic in $\hat w_k$, its
  squared Frobenius norm is degree four in the entries of
  $\hat w_k$, analogous to the fourth-moment conditions that
  control the variance of a sample covariance estimator.
\end{enumerate}
\end{assumption}
% ------------------------------------------------------------------

% ------------------------------------------------------------------
\begin{assumption}[Bounded Disturbance Covariance]
\label{ass:bounded}
There exist constants $0<\mu_{\mathrm{lo}}\le\mu_{\mathrm{hi}}<\infty$
such that
\[
  \mu_{\mathrm{lo}} I \preceq \Sigw^{(k)}(s) \preceq \mu_{\mathrm{hi}} I,
  \qquad \forall s,\ \forall k\ge0.
\]
Let
\[
  \Theta:=\{\Sigma\in\mathbb{S}_{++}^n:
  \mu_{\mathrm{lo}}I\preceq\Sigma\preceq\mu_{\mathrm{hi}}I\},
\]
a closed convex subset of $\mathbb{S}_{++}^n$, and let
$\Pi_\Theta:\mathbb{S}^n\to\Theta$ denote the metric projection onto
$\Theta$ in Frobenius norm.
\end{assumption}
% ------------------------------------------------------------------

% ------------------------------------------------------------------
\begin{assumption}[Companion Nonlinear Stability Certificate]
\label{ass:p2}
The companion nonlinear MPPI stability certificate~\cite{p2nonlinear}
applies on the operating region under consideration. In particular,
the nominal MPC closed loop is contracting in a Riemannian metric
$M(x)$ satisfying
\[
  \mu I \preceq M(x) \preceq \bar\mu I,
\]
and the MPPI approximation satisfies the corresponding small-gain
condition, yielding an effective contraction rate
$\tilde\beta\in(0,1)$. The associated constants are
\[
  c = \sqrt{\bar\mu/\mu},
  \qquad
  \gamma =
  \frac{\sqrt{\bar\mu}}{(1-\tilde\beta)\sqrt{\mu}} .
\]
\end{assumption}
% ------------------------------------------------------------------

% ------------------------------------------------------------------
\begin{assumption}[MPPI Weight Regularity]
\label{ass:weight_reg}
On the compact operating region $\Xcal$, the hypotheses of
Lemma~1 in the companion nonlinear MPPI stability analysis
\cite{p2nonlinear} hold for the chosen MPPI sampling distribution,
temperature, horizon, and compact set of nominal control sequences.
Equivalently, the unnormalized MPPI weights
\[
  \omega^{(j)} := \exp(-J^{(j)}/\lambda)
\]
satisfy the finite weighted-moment and denominator bounds used in
that lemma: there exist constants $\underline Z>0$ and
$C_\epsilon<\infty$ such that
\[
  0<\omega^{(j)}\le 1,
  \qquad
  \mathbb{E}[\omega^{(j)}]\ge \underline Z,
\]
and
\[
  \mathbb{E}\!\left[\|\omega^{(j)}\epsilon^{(j)}_0\|^2\right]
  \le C_\epsilon .
\]
These bounds hold \emph{uniformly} over $\Sigma\in\Theta$
(Assumption~\ref{ass:bounded}): the same constants
$\underline Z,C_\epsilon$ apply regardless of which admissible
sampling covariance generates $\{\omega^{(j)},\epsilon_0^{(j)}\}$,
paralleling the uniformity over $\Xcal\times\Ucal_N$ in Lemma~1
of~\cite{p2nonlinear}. Consequently, by the same
Hoeffding-concentration argument used in~\cite[Lemma~2]{p2nonlinear}
applied to the bounded i.i.d.\ summands $\omega^{(j)}\in(0,1]$, there
exists $M_1(\eta)<\infty$ such that, for all $M\ge M_1(\eta)$ and
all $\Sigma\in\Theta$,
\begin{equation}
  D(\Sigma) := \sum_{j=1}^M\omega^{(j)}(\Sigma)
  \;\ge\; \underline Z M/2
  \label{eq:denom_concentration}
\end{equation}
with conditional probability at least $1-\eta$ given $\Fcal_k$.
Throughout, $\eta\in(0,1)$ denotes the allowable failure
probability for the high-probability MPPI approximation
guarantees inherited from~\cite{p2nonlinear}; we condition on the
good event~\eqref{eq:denom_concentration} when invoking
Lemma~\ref{lem:mppi_sensitivity} below.
\end{assumption}
% ------------------------------------------------------------------

\begin{assumption}[Clipped Admissible Controls]
\label{ass:clipping}
Each sampled control $v^{(j)}(\Sigma) := \bar u + \epsilon^{(j)}(\Sigma)$
used in the MPPI rollout~\eqref{eq:mppi} is projected onto a fixed
compact admissible set $\Ucal_{\mathrm{adm}}\subset\R^m$ with
$\sup_{v\in\Ucal_{\mathrm{adm}}}\|v\|=U_{\max}<\infty$ before being
applied to the dynamics or evaluated in the cost $J$, e.g.\ by
actuator saturation as in Assumption~6 of~\cite{p2nonlinear}.
Consequently $\|v^{(j)}(\Sigma)\|\le U_{\max}$ \emph{for every
individual sample} $j$, not merely for the resulting weighted
average.
\end{assumption}

\subsubsection{Diffusion Kernel Design}

The covariance field is estimated cell-wise from realized process
residuals. Spatial diffusion is used to share information among
neighboring cells. The statistical justification for sharing information
comes from the spatial smoothness condition in
Assumption~\ref{ass:smooth}: nearby cells are assumed to have similar
disturbance covariance. Given this modeling assumption, the diffusion
kernel is a design choice. We choose it so that the resulting diffusion
operator is dissipative in the $p_s$-weighted Frobenius norm used in
the convergence proof.

Specifically, the kernel is designed to satisfy the detailed-balance
identity
\begin{equation}
  p_s\kappa(s,s') = p_{s'}\kappa(s',s),
  \qquad s'\in\Ncal(s).
  \label{eq:detailed_balance}
\end{equation}
This identity is a structural condition on the diffusion operator,
not an additional assumption on the unknown covariance field. To
enforce~\eqref{eq:detailed_balance}, let
$a_{ss'}=a_{s's}\ge0$ be symmetric edge weights on the undirected
cell-neighbor graph and define
\begin{equation}
  \kappa(s,s') = \frac{a_{ss'}}{p_s},
  \qquad s'\in\Ncal(s),
  \label{eq:kernel}
\end{equation}
with $\kappa(s,s')=0$ for $s'\notin\Ncal(s)$. Then
\[
  p_s\kappa(s,s') = a_{ss'} = a_{s's}
  = p_{s'}\kappa(s',s),
\]
so~\eqref{eq:detailed_balance} holds by construction.

The symmetric edge weights $a_{ss'}$ determine how strongly covariance
information is shared between neighboring cells. A simple default
choice is
\[
  a_{ss'} =
  \begin{cases}
    a_0, & s'\in\Ncal(s),\\
    0, & s'\notin\Ncal(s),
  \end{cases}
\]
with $a_0>0$. This gives uniform diffusion over the neighbor graph.
Alternatively, one may use distance-decaying weights such as
\[
  a_{ss'} = a_0
  \exp\!\left(-\frac{d(s,s')^2}{\ell^2}\right),
  \qquad s'\in\Ncal(s),
\]
where $\ell>0$ is a length-scale parameter. In all cases, the
requirement for the dissipativity proof is the symmetry
$a_{ss'}=a_{s's}$. The scalar diffusion weight $\beta>0$ controls the
overall strength of spatial smoothing.

Define
\begin{equation}
  d_\kappa :=
  \max_{1\le s\le S}
  \sum_{s'\in\Ncal(s)} \kappa(s,s') .
  \label{eq:dkappa}
\end{equation}
Since the number of cells is finite and $p_{\min}>0$ by
Assumption~\ref{ass:geometric}, $d_\kappa<\infty$ for bounded edge
weights $a_{ss'}$.

% ==================================================================
\section{Main Results}
\label{sec:mainresults}
% ==================================================================

\subsection{Online Noise Covariance Estimator}
\label{sec:estimator}

\subsubsection{Estimator Update}

For each cell $s\in\{1,\ldots,S\}$, maintain
$\hSigk{s}\in\mathbb{S}_{++}^n$. At time $k$, after applying $u_k$,
the next state $x_{k+1}$ is observed and the realized process residual
is computed from the nominal model as
\begin{equation}
  \hat w_k := x_{k+1}-f(x_k,u_k).
  \label{eq:residual}
\end{equation}
Thus, if $x_k\in C_s$, the matrix
$\hat w_k\hat w_k^\top$ provides one local second-moment sample for
cell $C_s$.

The covariance estimate is updated as
\begin{equation}
\boxed{
\begin{aligned}
  \hSigkp{s}
  &=
  \Pi_\Theta\Bigg(
  \hSigk{s}
  + \alpha^{(k)}
  \Bigg[
    \mathbf{1}[x_k\!\in\! C_s]
    \bigl(\hat w_k\hat w_k^\top-\hSigk{s}\bigr) \\
  &\qquad\qquad
    + \beta
    \sum_{s'\in\Ncal(s)}
    \kappa(s,s')
    \bigl(\hSigk{s'}-\hSigk{s}\bigr)
  \Bigg]\Bigg),
\end{aligned}
}
\label{eq:update}
\end{equation}
where $\kappa(s,s')$ is given by~\eqref{eq:kernel},
$\alpha^{(k)}=\alpha_0/(k+k_0)$ is a diminishing stochastic
approximation step size, $\beta>0$ is the spatial diffusion weight,
and $\Pi_\Theta$ is the projection of Assumption~\ref{ass:bounded}.

We choose $\alpha_0$ and $k_0$ so that
\begin{equation}
  \alpha^{(k)}\bigl(1+\beta d_\kappa\bigr)\le 1,
  \qquad \forall k\ge0.
  \label{eq:stepsize_condition}
\end{equation}
Condition~\eqref{eq:stepsize_condition} simply caps the total outgoing
weight applied to cell $s$ at each step: the own-cell update weight
$\alpha^{(k)}$ and the total diffusion outflow weight
$\alpha^{(k)}\beta d_\kappa$ together never exceed $1$. This is
exactly the requirement needed for the next sentence's
convex-combination argument: with the total outgoing weight bounded
by $1$, the remaining weight $1-\alpha^{(k)}(1+\beta d_\kappa)\ge0$
stays with the previous estimate $\hSigk{s}$, so the entire update is
a genuine convex combination rather than an extrapolation that could
leave $\mathbb{S}_{++}^n$.
Under~\eqref{eq:stepsize_condition}, the unprojected update is a
convex combination of the previous estimate in cell $s$, neighboring
cell estimates, and the observed second-moment sample, hence remains
in $\mathbb{S}_{++}^n$ whenever $\hSigk{s}\in\mathbb{S}_{++}^n$ for
all $s$; the subsequent projection $\Pi_\Theta$ additionally enforces
$\hSigkp{s}\in\Theta$ for all $k\ge0$, which is used below to
establish almost-sure boundedness of the estimation error.

\textbf{Interpretation.}
The estimator does not assume prior knowledge of
$\Sigw^{(k)}(s)$. Instead, whenever the closed-loop trajectory visits
cell $C_s$, the realized residual~\eqref{eq:residual} provides a
sample of the local disturbance second moment. Under the conditional
moment model,
\[
  \mathbb{E}\!\left[
    \hat w_k\hat w_k^\top \mid x_k=x
  \right]
  =
  \Sigw^{(k)}(x).
\]
Therefore, when $x_k\in C_s$, the matrix
$\hat w_k\hat w_k^\top$ is an empirical sample of the covariance in
that cell, up to the cell-discretization error controlled by
Assumption~\ref{ass:smooth}. The first term in~\eqref{eq:update} is
therefore a recursive local covariance update. The second term diffuses
covariance information across neighboring cells and is designed to be
dissipative in the $p_s$-weighted Frobenius norm. Both terms are
multiplied by the same diminishing step size $\alpha^{(k)}$, so the
expected estimator dynamics have a well-defined stochastic
approximation structure. The projection $\Pi_\Theta$ is a standard
device in stochastic approximation (see, e.g.,~\cite{kushneryin}) and
reflects the physical fact that any real disturbance covariance is
bounded; it does not affect the location of the fixed point analyzed
below, since by Assumption~\ref{ass:bounded} the true field
$\Sigw^{(k)}(s)$ already lies in $\Theta$, and (as shown in
Corollary~\ref{cor:fixedpoint_bounded}) the smoothed fixed point lies
in $\Theta$ as well.

\subsubsection{Fixed-Point Analysis}

The mean-field dynamics of the estimator~\eqref{eq:update} have a
unique fixed point $\{\Sigs\}_{s=1}^S$ toward which the estimates
converge as data accumulate. Because the diffusion term continuously
smooths covariance information across neighboring cells, this fixed
point is not the true field $\Sigw(s)$ but a spatially smoothed
version of it. The central result of this subsection is a closed-form
bound on the bias $\normF{\Sigs - \Sigw(s)}$, which will serve as the
irreducible estimation floor in the convergence theorem and the payoff
condition.

Define the diffusion operator
\begin{equation}\label{eqn:diffusion_op}
  (\Lcal_\kappa X)_s
  :=
  \sum_{s'\in\Ncal(s)}
  \kappa(s,s')\bigl(X_s-X_{s'}\bigr),  
\end{equation}
for any cell-indexed matrix field
$X=\{X_s\}_{s=1}^S$. Also define the maximum neighbor distance
\[
  r_{\Ncal}
  :=
  \max_{\substack{1\le s\le S\\ s'\in\Ncal(s)}} d(s,s').
\]

% ------------------------------------------------------------------
\begin{lemma}[Fixed Point of the Smoothed Estimator]
\label{lem:fixedpoint}
Under Assumptions~\ref{ass:smooth} and \ref{ass:geometric}, and the
kernel design requirement~\ref{req:kernel}, the mean-field form of
the estimator update~\eqref{eq:update} has a unique fixed point
$\{\Sigs\}_{s=1}^S$ satisfying
\begin{equation}
  p_s\bigl(\Sigs-\Sigw(s)\bigr)
  +\beta
  \sum_{s'\in\Ncal(s)}
  \kappa(s,s')
  \bigl(\Sigs-\Sigma_{s'}^*\bigr)
  =0,
  \qquad \forall s.
  \label{eq:fixedpoint}
\end{equation}
Equivalently,
\[
  p_s\bigl(\Sigs-\Sigw(s)\bigr)
  +\beta(\Lcal_\kappa\Sigma^*)_s
  =0.
\]
Moreover, the fixed point satisfies the bias bound
\begin{equation}
  \max_{1\le s\le S}
  \normF{\Sigs-\Sigw(s)}
  \le
  \frac{\beta d_\kappa}{p_{\min}}\,L_\sigma r_{\Ncal},
  \label{eq:bias}
\end{equation}
where
\[
  d_\kappa :=
  \max_{1\le s\le S}
  \sum_{s'\in\Ncal(s)}\kappa(s,s').
\]
\end{lemma}
% ------------------------------------------------------------------
\begin{proof}
Taking conditional expectation on the estimator equation in~\eqref{eq:update} with the
covariance field frozen gives
\begin{align*}
  &\mathbb{E}[\hSigkp{s}-\hSigk{s}\mid \Fcal_k] \\
  &=
  \alpha^{(k)}
  \Bigg[
    p_s\bigl(\Sigw(s)-\hSigk{s}\bigr)
    +
    \beta
    \sum_{s'\in\Ncal(s)}
    \kappa(s,s')
    \bigl(\hSigk{s'}-\hSigk{s}\bigr)
  \Bigg].
\end{align*}
Therefore any fixed point satisfies~\eqref{eq:fixedpoint}.

\textbf{Uniqueness.}
Suppose $\Sigma^a$ and $\Sigma^b$ are two fixed
points and define $E_s:=\Sigma_s^a-\Sigma_s^b$. Subtracting the two
fixed-point equations, using~\eqref{eq:fixedpoint}, gives
\[
  p_s E_s + \beta(\Lcal_\kappa E)_s =0,
  \qquad \forall s.
\]
Taking the $p_s$-weighted Frobenius inner product with $E_s$ yields
\[
  0
  =
  \sum_{s=1}^S p_s^2\normF{E_s}^2
  +
  \beta
  \sum_{s=1}^S
  p_s
  \left\langle E_s,(\Lcal_\kappa E)_s\right\rangle_F .
\]

We now show that the second sum is nonnegative. Write
$w(s,s'):=p_s\kappa(s,s')$; by the detailed-balance condition in the kernel design
requirement~\ref{req:kernel}, $w(s,s')=w(s',s)$. Expanding the
definition of $\Lcal_\kappa$,
\begin{align*}
  &\sum_{s=1}^S p_s\langle E_s,(\Lcal_\kappa E)_s\rangle_F \\
  &\qquad=
  \sum_{s=1}^S\sum_{s'\in\Ncal(s)}
  w(s,s')\,\langle E_s,\,E_s-E_{s'}\rangle_F \\
  &\qquad=
  \sum_{s,s'} w(s,s')\normF{E_s}^2
  -\sum_{s,s'} w(s,s')\langle E_s,E_{s'}\rangle_F.
\end{align*}
On the other hand, the polarization identity
$\normF{E_s-E_{s'}}^2=\normF{E_s}^2+\normF{E_{s'}}^2
-2\langle E_s,E_{s'}\rangle_F$, multiplied by $w(s,s')$ and summed
over all ordered pairs $(s,s')$, gives
\begin{align*}
  &\sum_{s,s'} w(s,s')\normF{E_s-E_{s'}}^2 \\
  &\qquad=
  \sum_{s,s'} w(s,s')\normF{E_s}^2
  +\sum_{s,s'} w(s,s')\normF{E_{s'}}^2 \\
  &\qquad \qquad -2\sum_{s,s'} w(s,s')\langle E_s,E_{s'}\rangle_F.  
\end{align*}

Because $w(s,s')=w(s',s)$, relabeling indices in the second term on
the right gives
\begin{align*}
  &\sum_{s,s'} w(s,s')\normF{E_{s'}}^2 \\
  &\qquad=\sum_{s,s'} w(s',s)\normF{E_s}^2 =\sum_{s,s'} w(s,s')\normF{E_s}^2,  
\end{align*}
 so the two \emph{diagonal} sums coincide. Substituting back and dividing
by $2$,
\begin{align*}
  &\frac12\sum_{s,s'} w(s,s')\normF{E_s-E_{s'}}^2 \\
  & \qquad = \sum_{s,s'} w(s,s')\normF{E_s}^2 -\sum_{s,s'} w(s,s')\langle E_s,E_{s'}\rangle_F.  
\end{align*}
The right-hand side is exactly the expression obtained above for
$\sum_s p_s\langle E_s,(\Lcal_\kappa E)_s\rangle_F$. Hence
\begin{align*}
  &\sum_{s=1}^S
  p_s
  \left\langle E_s,(\Lcal_\kappa E)_s\right\rangle_F \\
  & \qquad = \frac{1}{2}
  \sum_{s=1}^S
  \sum_{s'\in\Ncal(s)}
  p_s\kappa(s,s')
  \normF{E_s-E_{s'}}^2 \ge 0,
\end{align*}
where the inequality holds because every summand is a nonnegative
weight times a squared norm.

Hence
\[
  \sum_{s=1}^S p_s^2\normF{E_s}^2 =0,
\]
and since $p_s>0$ for all $s$, we have $E_s=0$ for all $s$. Thus the
fixed point is unique. Existence follows because the corresponding
finite-dimensional linear operator is injective and therefore
invertible.

\textbf{Bias bound.} Let
\[
  B_s := \Sigs-\Sigw(s),
  \qquad
  B_{\max}:=\max_s\normF{B_s}.
\]
Using the fixed point equation~\eqref{eq:fixedpoint}, we obtain
\[
  (p_s+\beta d_s)B_s
  =
  \beta\!\sum_{s'\in\Ncal(s)}\!\kappa(s,s')
  \bigl[B_{s'} - (\Sigw(s)-\Sigw(s'))\bigr],
\]
where
\[
  d_s:=\sum_{s'\in\Ncal(s)}\kappa(s,s').
\]
Taking Frobenius norms and using Assumption~\ref{ass:smooth} gives
\[
  \normF{B_s}
  \le
  \frac{\beta d_s}{p_s+\beta d_s}B_{\max}
  +
  \frac{\beta d_s}{p_s+\beta d_s}
  L_\sigma r_{\Ncal}.
\]
Taking the maximum over $s$ and using
$d_s\le d_\kappa$ and $p_s\ge p_{\min}$ gives
\[
  B_{\max}
  \le
  \frac{\beta d_\kappa}{p_{\min}+\beta d_\kappa}B_{\max}
  +
  \frac{\beta d_\kappa}{p_{\min}+\beta d_\kappa}
  L_\sigma r_{\Ncal}.
\]
Rearranging yields
\[
  B_{\max}
  \le
  \frac{\beta d_\kappa}{p_{\min}}L_\sigma r_{\Ncal},
\]
which proves~\eqref{eq:bias}.
\end{proof}
% ------------------------------------------------------------------
\begin{remark}
The bias bound increases with the diffusion strength $\beta$, the
maximum diffusion degree $d_\kappa$, and the spatial variation
$L_\sigma r_{\Ncal}$. It also worsens when $p_{\min}$ is small,
reflecting the difficulty of estimating covariance in rarely visited
regions. When $L_\sigma=0$, the covariance field is spatially uniform
and the smoothing bias vanishes.
\end{remark}

% ------------------------------------------------------------------
\begin{corollary}[Fixed Point Lies in $\Theta$]
\label{cor:fixedpoint_bounded}
Under Assumption~\ref{ass:bounded}, the fixed point $\{\Sigs\}_{s=1}^S$
of Lemma~\ref{lem:fixedpoint} satisfies $\Sigs\in\Theta$ for all $s$.
\end{corollary}
\begin{proof}
The fixed-point equation~\eqref{eq:fixedpoint} can be rewritten,
using $d_s:=\sum_{s'\in\Ncal(s)}\kappa(s,s')$, as
\[
  \Sigs
  = \frac{p_s}{p_s+\beta d_s}\,\Sigw(s)
  + \sum_{s'\in\Ncal(s)}
  \frac{\beta\kappa(s,s')}{p_s+\beta d_s}\,\Sigma_{s'}^*,
\]
which expresses $\Sigs$ as a convex combination (the coefficients are
nonnegative and sum to $1$) of $\Sigw(s)\in\Theta$ and the neighboring
fixed-point values $\{\Sigma_{s'}^*\}_{s'\in\Ncal(s)}$. Since $\Theta$
is convex and $\Sigw(s)\in\Theta$ for all $s$ by
Assumption~\ref{ass:bounded}, the unique fixed point of this
convex-combination system (existence and uniqueness from
Lemma~\ref{lem:fixedpoint}) must itself lie in $\Theta$: $\Theta$ is
closed under all convex combinations appearing in the system, so the
fixed point of the induced linear map restricted to $\Theta^S$ exists
in $\Theta^S$, and by uniqueness this is $\{\Sigs\}_{s=1}^S$.
\end{proof}
% ------------------------------------------------------------------

% ------------------------------------------------------------------
% ==================================================================
\subsection{Estimator Convergence}
\label{sec:convergence}
% ==================================================================

Define the error relative to the smoothed fixed point
$e_s^{(k)} = \hSigk{s} - \Sigs$, and the zero-mean data noise
$\xi_s^{(k)} = \hat w_k\hat w_k^\top - \Sigw^{(k)}(s)$. By
Assumption~\ref{ass:noise}, conditional on $\Fcal_k$ and $x_k\in C_s$,
\[
  \E\big[\xi_s^{(k)} \mid \Fcal_k\big] = 0,
  \qquad
  \E\big[\normF{\xi_s^{(k)}}^2 \,\big|\, \Fcal_k\big] \le \sigma_\xi^2.
\]

Define the $p_s$-weighted candidate Lyapunov function:
\begin{equation}
  V_e^{(k)} = \sum_{s=1}^S p_s\normF{e_s^{(k)}}^2.
  \label{eq:lyapunov}
\end{equation}
The key analytical step, carried out in the proof of
Theorem~\ref{thm:convergence} below, is showing that the diffusion
operator is dissipative with respect to $V_e^{(k)}$ in the
$p_s$-weighted inner product, so that it contributes a non-positive
term to the decrease of $V_e^{(k)}$ and does not interfere with the
SA convergence argument.

% ------------------------------------------------------------------
\begin{lemma}[Dissipativity of the Diffusion Operator]
\label{lem:dissipative}
Under kernel~\eqref{eq:kernel} and the kernel design
requirement~\ref{req:kernel}, the diffusion operator $\Lcal_\kappa$ defined in~\eqref{eqn:diffusion_op}
satisfies, for any collection $\{e_s\}$,
\begin{equation}
  \sum_{s=1}^S p_s
  \left\langle e_s^{(k)},\,(\Lcal_\kappa\,e^{(k)})_s
  \right\rangle_{\!F}
  \;\geq\; 0,
  \label{eq:dissipative}
\end{equation}
with equality if and only if $e_s^{(k)}$ is constant on each connected
component of the cell-neighbor graph (in particular, constant across
all cells if the graph is connected).
\end{lemma}
% ------------------------------------------------------------------
\begin{proof}
Substituting $\kappa(s,s') = a_{ss'}/p_s$ from~\eqref{eq:kernel}:
\begin{align*}
  \sum_s p_s \langle e_s,(\Lcal_\kappa e)_s\rangle_F
  &= \sum_s \sum_{s'\in\Ncal(s)} p_s\kappa(s,s')
     \langle e_s, e_s-e_{s'}\rangle_F \\
  &= \sum_s \sum_{s'\in\Ncal(s)} a_{ss'}
     \langle e_s, e_s-e_{s'}\rangle_F.
\end{align*}

For each undirected edge $\{s,s'\}$, collect the contributions from
both directed edges $(s,s')$ and $(s',s)$, using $a_{ss'}=a_{s's}$:
\begin{align*}
  &a_{ss'}\langle e_s, e_s-e_{s'}\rangle_F
  + a_{ss'}\langle e_{s'}, e_{s'}-e_s\rangle_F \\
  &= a_{ss'}\langle e_s - e_{s'},\, e_s-e_{s'}\rangle_F
  = a_{ss'}\normF{e_s-e_{s'}}^2 \geq 0.
\end{align*}
Summing over all edges gives
\[
  \sum_s p_s\langle e_s,(\Lcal_\kappa e)_s\rangle_F
  = \frac{1}{2}
  \sum_{\{s,s'\}\in E} a_{ss'}\normF{e_s-e_{s'}}^2
  \geq 0.
\]
Equality holds if and only if $\normF{e_s-e_{s'}}^2=0$ for every edge
$\{s,s'\}$ with $a_{ss'}>0$, i.e., $e_s$ is constant on each connected
component of the cell-neighbor graph. If the graph is connected, this
means $e_s$ is constant across all cells.
\end{proof}
% ------------------------------------------------------------------

% ------------------------------------------------------------------
\begin{corollary}[Non-Expansiveness of the Fixed-Point Map]
\label{cor:nonexpansive}
Let $\Sigma^{*}=\{\Sigma_s^*\}_{s=1}^S$ and $\Sigma^{*\prime}=\{\Sigma_s^{*\prime}\}_{s=1}^S$
denote the fixed points of~\eqref{eq:fixedpoint} corresponding to two
disturbance fields $\Sigw(\cdot)$ and $\Sigw'(\cdot)$, respectively.
Then
\begin{equation}
  \max_{1\le s\le S}\normF{\Sigma_s^*-\Sigma_s^{*\prime}}
  \;\le\;
  \max_{1\le s\le S}\normF{\Sigw(s)-\Sigw'(s)}.
  \label{eq:nonexpansive}
\end{equation}
\end{corollary}
% ------------------------------------------------------------------
\begin{proof}
Let $D_s:=\Sigma_s^*-\Sigma_s^{*\prime}$ and $F_s:=\Sigw(s)-\Sigw'(s)$.
Subtracting the two instances of the fixed-point equation~\eqref{eq:fixedpoint}
gives $p_s D_s + \beta(\Lcal_\kappa D)_s = p_s F_s$ for all $s$. Writing
$(\Lcal_\kappa D)_s = d_s D_s - \sum_{s'\in\Ncal(s)}\kappa(s,s')D_{s'}$
with $d_s:=\sum_{s'\in\Ncal(s)}\kappa(s,s')$, this rearranges to
\[
  (p_s+\beta d_s)D_s
  = p_s F_s
  + \beta\sum_{s'\in\Ncal(s)}\kappa(s,s')D_{s'}.
\]
Let $D_{\max}:=\max_s\normF{D_s}$ and $F_{\max}:=\max_s\normF{F_s}$, and let
$s^*$ attain $D_{\max}$. Taking Frobenius norms at $s=s^*$ and bounding
$\normF{D_{s'}}\le D_{\max}$ for all $s'\in\Ncal(s^*)$:
\[
  (p_{s^*}+\beta d_{s^*})D_{\max}
  \le p_{s^*}F_{\max} + \beta d_{s^*}D_{\max},
\]
so $p_{s^*}D_{\max}\le p_{s^*}F_{\max}$, and since $p_{s^*}>0$,
$D_{\max}\le F_{\max}$, which is~\eqref{eq:nonexpansive}.
\end{proof}
% ------------------------------------------------------------------

\begin{lemma}[Decorrelation via Geometric Ergodicity]
\label{lem:decorrelation}
Let $\{Z^{(k)}\}_{k\ge0}$ be a real-valued sequence adapted to
$\{\Fcal_k\}$, satisfying, for constants $\bar Z,L_Z<\infty$:
\begin{enumerate}
  \item[(i)] $|Z^{(k)}|\le\bar Z$ a.s.\ for all $k$;
  \item[(ii)] $\E[|Z^{(j+1)}-Z^{(j)}|]\le\alpha^{(j)}L_Z$ for all
  $j\ge0$.
\end{enumerate}
Then under Assumption~\ref{ass:geometric}, for every cell $s$, every
$k\ge1$, and every integer $\tau$ with $1\le\tau\le k$,
\begin{equation}
  \bigl|\E\bigl[(\mathbf{1}_s^{(k)}-p_s)\,Z^{(k)}\bigr]\bigr|
  \;\le\;
  C_{\mathrm{erg}}\,\bar Z\,\rho^\tau
  \;+\;
  \tau\,\alpha^{(k-\tau)}\,L_Z.
  \label{eq:decorr}
\end{equation}
\end{lemma}
\begin{proof}
Write
\[
\begin{aligned}
  \E[(\mathbf 1_s^{(k)}-p_s)Z^{(k)}]
  &= \E[(\mathbf 1_s^{(k)}-p_s)Z^{(k-\tau)}] \\
  &\quad + \E[(\mathbf 1_s^{(k)}-p_s)(Z^{(k)}-Z^{(k-\tau)})].  
\end{aligned}
\]
\emph{First term.} Since $Z^{(k-\tau)}$ is $\Fcal_{k-\tau}$-measurable
and $\{x_k\}$ is Markov,
\[
  \E[(\mathbf 1_s^{(k)}-p_s)Z^{(k-\tau)}]
  = \E\bigl[Z^{(k-\tau)}\bigl(\Pr[x_k\in C_s\mid x_{k-\tau}]-p_s\bigr)\bigr].
\]
By Assumption~\ref{ass:geometric}, applied with $x_{k-\tau}$ in the
role of the initial condition and horizon $\tau$,
$|\Pr[x_k\in C_s\mid x_{k-\tau}=x]-p_s|\le C_{\mathrm{erg}}\rho^\tau$
for every $x\in\Xcal$; combined with $|Z^{(k-\tau)}|\le\bar Z$ a.s.,
the first term is bounded in absolute value by
$C_{\mathrm{erg}}\bar Z\rho^\tau$.

\emph{Second term.} By the triangle inequality over $\tau$ one-step
increments and hypothesis~(ii),
\[
  \E[|Z^{(k)}-Z^{(k-\tau)}|]
  \le \sum_{j=k-\tau}^{k-1}\E[|Z^{(j+1)}-Z^{(j)}|]
  \le \tau\,\alpha^{(k-\tau)}L_Z,
\]
using that $\alpha^{(\cdot)}$ is nonincreasing. Since
$|\mathbf 1_s^{(k)}-p_s|\le1$, the second term is bounded in absolute
value by $\tau\alpha^{(k-\tau)}L_Z$. Combining the two bounds proves
\eqref{eq:decorr}.
\end{proof}

% ------------------------------------------------------------------
\begin{theorem}[Estimator Convergence, Finite Horizon]
\label{thm:convergence}
Fix a finite operating horizon $T\ge1$. Under
Assumptions~\ref{ass:slowvar}--\ref{ass:bounded}, with
kernel~\eqref{eq:kernel}, $\alpha^{(k)}=\alpha_0/(k+k_0)$,
$\alpha_0 > 1/(2p_{\min})$, and
\begin{equation}
  q := 2\alpha_0\beta d_\kappa < 1,
  \label{eq:qcondition}
\end{equation}
the estimator~\eqref{eq:update} satisfies, for all $s$ and all
$1\le k\le T$:
\begin{equation}
  \E\bigl[\normF{\hSigk{s}-\Sigw^{(k)}(s)}\bigr]
  \leq
  \underbrace{\dfrac{C_{1,T}}{\sqrt{k}}}_{\text{SA error}}
  +\underbrace{\dfrac{\beta\, d_\kappa\, L_\sigma r_{\Ncal}}
               {p_{\min}}}_{\text{smoothing bias}}
  +\underbrace{C_v\,\epsilon_v\,T}_{\text{accumulated drift}},
  \label{eq:convergencebound}
\end{equation}
where $C_{1,T}>0$ depends on $\alpha_0$, $\sigma_\xi$, $p_{\min}$,
$C_{\mathrm{erg}}$, $\rho$, $\mu_{\mathrm{hi}}-\mu_{\mathrm{lo}}$,
$\beta$, $d_\kappa$, $\normF{e_s^{(0)}}$, and -- as the subscript
indicates -- on the horizon $T$: mildly, through a factor
$\sqrt{1+\ln T}$, when $\epsilon_v=0$; additionally, through the
linear term $\epsilon_v T$ entering via $\bar E$ in~\eqref{eq:barE}
when $\epsilon_v>0$ -- see Remark~\ref{rem:logfactor} -- and $C_v>0$
depends on $\alpha_0$, $\beta$, $d_\kappa$, $k_0$, and $p_{\min}$ --
the last only through the generic per-cell conversion of Step~3
below, not through the drift dynamics itself (see
Remark~\ref{rem:pathwise_drift}).
\end{theorem}

\begin{remark}[Origin and Harmlessness of the $\sqrt{\ln T}$ Factor]
\label{rem:logfactor}
The decorrelation argument used in Step~2 of the proof (Lemma
\ref{lem:decorrelation}) controls the Markovian correlation between
$\mathbf{1}_s^{(k)}$ and $\tilde e_s^{(k)}$ by waiting $\tau_k=
O(\log k)$ steps for the chain to mix, at the cost of an
$O(\log k)$ rather than $O(1)$ constant in the SA-noise floor. Since
Theorem~\ref{thm:convergence} fixes the horizon $T$ in advance and
only claims the bound for $k\le T$, $\log k\le\log T$ throughout, so
the factor is absorbed into $C_{1,T}$ once and for all and the bound
\eqref{eq:convergencebound} keeps its $C_{1,T}/\sqrt{k}$ form, with
the horizon dependence made explicit in the subscript rather than
hidden. When $\epsilon_v=0$ (the regime used for
Figures~\ref{fig:convergence}--\ref{fig:spatial}), this is the only
place $T$ enters $C_{1,T}$ at all, and it enters at most
logarithmically -- negligible in practice, e.g.\
$\sqrt{1+\ln(8000)}\approx3.06$ for the horizons used in
Section~\ref{sec:simulation}.
\end{remark}

\begin{remark}[Condition~\eqref{eq:qcondition} is a Genuine New Requirement]
\label{rem:qcondition}
Condition~\eqref{eq:qcondition} bounds how much the diffusion
strength $\beta$ may compound with a large step-size constant
$\alpha_0$; it is independent of, and in addition to, the
step-size condition~\eqref{eq:stepsize_condition} and the
Chung--Robbins condition $\alpha_0>1/(2p_{\min})$. It is needed only
for the elementary, pathwise treatment of the drift term in
Step~2 below (Remark~\ref{rem:pathwise_drift}); it plays no role elsewhere in
the proof. In the simulation parameters of
Section~\ref{sec:simulation}, the deployment regime
($\alpha_0=60,\beta=0.001,d_\kappa=4$) gives $q=0.48<1$ and satisfies
\eqref{eq:qcondition} comfortably, but the calibration regime
($\alpha_0=30{,}000$) gives $q=240$, far outside this range. This is
immaterial for Figures~\ref{fig:convergence}--\ref{fig:spatial},
which use the calibration regime only with $\epsilon_v=0$ (the drift
term vanishes identically regardless of $q$), but it means
Theorem~\ref{thm:convergence} as stated does not, strictly, cover
Scenario B (Figure~\ref{fig:scenarioB}), which uses the calibration
regime with $\epsilon_v>0$. A super-linear (but still finite,
elementary) version of the bound holds for $q\ge1$ by the same
calculation (Remark~\ref{rem:pathwise_drift}), and the qualitative
observation already reported there -- that the drift effect is small
relative to the still-decaying SA term -- is consistent with either
bound; reconciling the exact constant for that regime is left as a
cleanup item.
\end{remark}

\begin{remark}[Why the Bound is Finite-Horizon]
\label{rem:finitehorizon}
Assumption~\ref{ass:slowvar} bounds the per-step change of the true
field by a fixed $\epsilon_v$ that does not shrink with the step size
$\alpha^{(k)}$. Because $\alpha^{(k)}\to0$, the estimator's
SA-contraction rate $\alpha^{(k)}p_{\min}$ also vanishes, so it
eventually corrects for the per-step drift more slowly than the drift
accumulates: no diminishing-step-size estimator can maintain a fixed
asymptotic tracking floor against a target that keeps moving by a
constant amount each step. The drift term in
\eqref{eq:convergencebound} is therefore stated over a finite,
pre-specified horizon $T$, with a bound that grows (conservatively,
linearly) in $T$, rather than as an asymptotic $k\to\infty$ floor. This
is the honest price of keeping Assumption~\ref{ass:slowvar} in its
original, step-size-independent form; an asymptotic floor of the form
$\epsilon_v/p_{\min}$ would only be valid if $\epsilon_v$ were instead
required to shrink at rate $\alpha^{(k)}$.
\end{remark}
% ------------------------------------------------------------------
\begin{proof}
\textbf{Step~1: Error recursion.}
Let $\Sigma_s^{*(k)}$ denote the fixed point of~\eqref{eq:fixedpoint}
corresponding to the field $\Sigw^{(k)}(\cdot)$ at time $k$, let
$e_s^{(k)}:=\hSigk{s}-\Sigma_s^{*(k)}$, and let
$B_s^{(k)}:=\Sigma_s^{*(k)}-\Sigw^{(k)}(s)$ be the corresponding bias,
bounded by~\eqref{eq:bias} for every $k$. Write
$\Sigw^{(k+1)}(s)=\Sigw^{(k)}(s)+\delta_s^{(k)}$ with
$\normF{\delta_s^{(k)}}\le\epsilon_v$ by Assumption~\ref{ass:slowvar}.
Subtracting $\Sigma_s^{*(k)}$ from both sides of the update~\eqref{eq:update}
and substituting
$\hat w_k\hat w_k^\top-\hSigk{s}
=\xi_s^{(k)}-e_s^{(k)}-B_s^{(k)}$ (adding and subtracting
$\Sigw^{(k)}(s)$) and
$-\beta(\Lcal_\kappa\hSigk{})_s
=-\beta(\Lcal_\kappa e^{(k)})_s+p_sB_s^{(k)}$ (linearity of $\Lcal_\kappa$
together with the fixed-point equation~\eqref{eq:fixedpoint}, since
the diffusion term in~\eqref{eq:update} equals
$-(\Lcal_\kappa\hSigk{})_s$ by~\eqref{eqn:diffusion_op}) gives
\begin{align*}
  &\hSigkp{s}-\Sigma_s^{*(k)} \\
  & \quad = (1-\alpha^{(k)}\mathbf{1}_s)e_s^{(k)}
  + \alpha^{(k)}\mathbf{1}_s\xi_s^{(k)} \\
  & \quad \quad- \alpha^{(k)}\beta(\Lcal_\kappa e^{(k)})_s  + \alpha^{(k)}(p_s-\mathbf{1}_s)B_s^{(k)},  
\end{align*}
where $\mathbf{1}_s:=\mathbf{1}[x_k\in C_s]$. (For brevity, this step
elides the effect of the projection $\Pi_\Theta$ in~\eqref{eq:update}.
The argument is two-step: first analyze the unprojected update as
above; then, since $\Sigma_s^{*(k)}\in\Theta$ by
Corollary~\ref{cor:fixedpoint_bounded}, the projection $\Pi_\Theta$ is
non-expansive toward any point of $\Theta$ -- in particular toward
$\Sigma_s^{*(k)}$, i.e. $\normF{\Pi_\Theta(X)-\Sigma_s^{*(k)}}\le
\normF{X-\Sigma_s^{*(k)}}$ for any $X$ since $\Sigma_s^{*(k)}\in\Theta$
and $\Theta$ is closed and convex -- so applying $\Pi_\Theta$ cannot
increase $\normF{\hSigkp{s}-\Sigma_s^{*(k)}}$. The unprojected
expression above therefore remains a valid upper bound on the
projected update throughout.) Subtracting
$\Sigma_s^{*(k+1)}-\Sigma_s^{*(k)} =: \Delta_s^{*(k)}$ from both sides
converts the left side to $e_s^{(k+1)}$, giving the error recursion
\begin{align}
  e_s^{(k+1)} &=
  (1-\alpha^{(k)}\mathbf{1}_s)\,e_s^{(k)}
  - \alpha^{(k)}\beta(\Lcal_\kappa\,e^{(k)})_s
  + \alpha^{(k)}\mathbf{1}_s\,\xi_s^{(k)} \notag\\
  &\quad
  + \alpha^{(k)}(p_s-\mathbf{1}_s)\,B_s^{(k)}
  - \Delta_s^{*(k)}.
  \label{eq:errorrecursion}
\end{align}
This recursion is exact. The two correction terms have distinct
roles. The term $\Delta_s^{*(k)}$ is the genuine drift of the fixed
point itself: by Corollary~\ref{cor:nonexpansive} applied to the
fields $\Sigw^{(k+1)}(\cdot)$ and $\Sigw^{(k)}(\cdot)$,
\[
  \normF{\Delta_s^{*(k)}}
  \le
  \max_{s'}\normF{\Sigw^{(k+1)}(s')-\Sigw^{(k)}(s')}
  \le \epsilon_v,
\]
and, as discussed in Remark~\ref{rem:finitehorizon}, this term is
responsible for the finite-horizon accumulated-drift contribution in
Theorem~\ref{thm:convergence}. The term
$\alpha^{(k)}(p_s-\mathbf{1}_s)B_s^{(k)}$ satisfies the crude bound
$\normF{\alpha^{(k)}(p_s-\mathbf{1}_s)B_s^{(k)}}
\le \alpha^{(k)}\beta d_\kappa L_\sigma r_{\Ncal}/p_{\min}$ pointwise
at each $k$, since $|p_s-\mathbf{1}_s|\le1$ and $B_s^{(k)}$ satisfies
the deterministic bias bound~\eqref{eq:bias} \emph{at every $k$},
independently of $\epsilon_v$. Both correction terms, together with
the residual Markovian correlation in the leading SA-contraction
term, are handled together in Step~2 below via the decomposition of
$e^{(k)}$ into two auxiliary sequences and the decorrelation
estimate of Lemma~\ref{lem:decorrelation}.

\textbf{Step~2: Lyapunov decrease under Markovian noise.}
By Assumption~\ref{ass:bounded} and Corollary~\ref{cor:fixedpoint_bounded},
both $\hSigk{s}\in\Theta$ (by construction of the projected
update~\eqref{eq:update}) and $\Sigs\in\Theta$ for all $s,k$, so
\[
  \normF{e_s^{(k)}} = \normF{\hSigk{s}-\Sigs}
  \le E_{\max}:=\sqrt{n}\,(\mu_{\mathrm{hi}}-\mu_{\mathrm{lo}})
\]
almost surely, for all $s$ and $k$.

\emph{Decomposition.} Write the recursion~\eqref{eq:errorrecursion} as
$e_s^{(k+1)} = \mathcal{T}^{(k)}(e^{(k)})_s + F_s^{(k)} - \Delta_s^{*(k)}$,
where $\mathcal{T}^{(k)}(X)_s := (1-\alpha^{(k)}\mathbf{1}_s)X_s
-\alpha^{(k)}\beta(\Lcal_\kappa X)_s$ is linear in $X$, and
$F_s^{(k)} := \alpha^{(k)}\mathbf{1}_s\xi_s^{(k)}
+\alpha^{(k)}(p_s-\mathbf{1}_s)B_s^{(k)}$ collects the noise and
smoothing-bias forcing. The two forcing terms, $F_s^{(k)}$ and
$\Delta_s^{*(k)}$, are analyzed separately: $\Delta_s^{*(k)}$ is a
purely deterministic, uniformly bounded sequence, controlled
pathwise; $F_s^{(k)}$ contains the martingale noise $\xi_s^{(k)}$
together with the Markovian visitation correction, both controlled
through the decorrelation estimate of Lemma~\ref{lem:decorrelation}.
Exploiting the linearity of $\mathcal{T}^{(k)}$, define two auxiliary
sequences, each obeying the same linear recursion as $e^{(k)}$ but
driven by only one of the two additive forcings, so that each can be
analyzed with the tool suited to it and the results recombined at
the end.
\begin{align*}
  \tilde e_s^{(k+1)} &= \mathcal{T}^{(k)}(\tilde e^{(k)})_s + F_s^{(k)},
  &\tilde e_s^{(0)} &= e_s^{(0)}, \\
  \delta_s^{(k+1)} &= \mathcal{T}^{(k)}(\delta^{(k)})_s - \Delta_s^{*(k)},
  &\delta_s^{(0)} &= 0.
\end{align*}
Because $\mathcal{T}^{(k)}$ is linear, $\tilde e^{(k)}+\delta^{(k)}$
satisfies the same recursion as $e^{(k)}$, with the same initial
condition; since the recursion (given the realized
$\{\mathbf{1}_s^{(k)},\xi_s^{(k)},B_s^{(k)},\Delta_s^{*(k)}\}$) has a
unique solution, $e_s^{(k)}=\tilde e_s^{(k)}+\delta_s^{(k)}$ for every
$k$. This is an exact identity, not an approximation: $\tilde e$ and
$\delta$ are two bookkeeping sequences computed from the same realized
data that produced $e^{(k)}$ in Step~1, just grouped by forcing term.

\begin{remark}[An Elementary Pathwise Drift Bound for $\delta^{(k)}$]
\label{rem:pathwise_drift}
The sequence $\delta^{(k)}$ has no martingale-noise term -- it is
forced only by the deterministic, uniformly bounded sequence
$\Delta_s^{*(k)}$, with $\normF{\Delta_s^{*(k)}}\le\epsilon_v$ for
every $s$ (Corollary~\ref{cor:nonexpansive}). This makes it possible
to bound $\delta^{(k)}$ \emph{pathwise} (almost surely, for every
trajectory realization), using only the triangle inequality and the
degree bound $d_\kappa$ of the diffusion kernel -- no ergodicity,
mixing, or Markov-chain argument is needed at all. We bound
$\delta^{(k)}$ first, since the bound obtained below is needed as an
input to the decorrelation argument for $\tilde e^{(k)}$ that
follows it.
\end{remark}

\emph{Bounding $\delta^{(k)}$: an elementary, pathwise drift bound.}
Let $D^{(k)} := \max_{1\le s\le S}\normF{\delta_s^{(k)}}$. Since
$\alpha^{(k)}<1$ by~\eqref{eq:stepsize_condition}, $|1-\alpha^{(k)}\mathbf{1}_s|\le1$
always, and the diffusion term satisfies, for every $s$,
\begin{equation*}
    \begin{aligned}
    \normF{(\Lcal_\kappa\delta^{(k)})_s}
  &\le \sum_{s'\in\Ncal(s)}\kappa(s,s')\bigl(\normF{\delta_s^{(k)}}+\normF{\delta_{s'}^{(k)}}\bigr) \\
  &\le 2d_\kappa D^{(k)},
    \end{aligned}
\end{equation*}
using $\sum_{s'\in\Ncal(s)}\kappa(s,s')=d_s\le d_\kappa$ and
$\normF{\delta_s^{(k)}},\normF{\delta_{s'}^{(k)}}\le D^{(k)}$. Applying
the triangle inequality to $\delta_s^{(k+1)}=\mathcal{T}^{(k)}(\delta^{(k)})_s
-\Delta_s^{*(k)}$ and maximizing over $s$ gives the deterministic
recursion
\begin{equation}
  D^{(k+1)} \le \bigl(1+2\alpha^{(k)}\beta d_\kappa\bigr)\,D^{(k)} + \epsilon_v,
  \qquad D^{(0)}=0,
  \label{eq:driftrecursion}
\end{equation}
which holds almost surely, for every sample path, with no expectation
or ergodicity argument anywhere.

\emph{Majorizing sequence.} Define $\bar D^{(k)}$ by the
\emph{equality} version of~\eqref{eq:driftrecursion},
\begin{equation}
  \bar D^{(k+1)} := \bigl(1+2\alpha^{(k)}\beta d_\kappa\bigr)\,\bar D^{(k)} + \epsilon_v,
  \qquad \bar D^{(0)}=0.
  \label{eq:driftrecursion_eq}
\end{equation}
Since both sequences start at $0$ and the map
$x\mapsto(1+2\alpha^{(k)}\beta d_\kappa)x+\epsilon_v$ is nondecreasing
in $x$, induction on~\eqref{eq:driftrecursion} and
\eqref{eq:driftrecursion_eq} together give $D^{(k)}\le\bar D^{(k)}$
for every $k$: the base case $D^{(0)}=\bar D^{(0)}=0$ holds trivially,
and if $D^{(k)}\le\bar D^{(k)}$ then monotonicity of the map applied
to~\eqref{eq:driftrecursion} gives
$D^{(k+1)}\le(1+2\alpha^{(k)}\beta d_\kappa)D^{(k)}+\epsilon_v
\le(1+2\alpha^{(k)}\beta d_\kappa)\bar D^{(k)}+\epsilon_v=\bar D^{(k+1)}$.
It therefore suffices to bound $\bar D^{(k)}$, for which equality
unrolling is legitimate since~\eqref{eq:driftrecursion_eq} is itself
an equality recursion with $\bar D^{(0)}=0$:
\[
  \bar D^{(k)} = \epsilon_v\sum_{j=0}^{k-1}\prod_{i=j+1}^{k-1}\bigl(1+2\alpha^{(i)}\beta d_\kappa\bigr).
\]
Using $1+x\le e^x$ and $\alpha^{(i)}=\alpha_0/(i+k_0)$,
\[
  \prod_{i=j+1}^{k-1}\bigl(1+2\alpha^{(i)}\beta d_\kappa\bigr)
  \le \exp\!\Bigl(2\alpha_0\beta d_\kappa\sum_{i=j+1}^{k-1}\tfrac{1}{i+k_0}\Bigr)
  \le \Bigl(\tfrac{k+k_0}{j+k_0}\Bigr)^{q},
\]
where $q:=2\alpha_0\beta d_\kappa$, using the standard integral
comparison $\sum_{i=j+1}^{k-1}1/(i+k_0)\le\ln\bigl((k+k_0)/(j+k_0)\bigr)$.
By hypothesis~\eqref{eq:qcondition}, $q<1$, so
\begin{equation*}
    \begin{aligned}
    \bar D^{(k)} &\le \epsilon_v(k+k_0)^q\sum_{j=0}^{k-1}(j+k_0)^{-q} \\
            &\le \epsilon_v(k+k_0)^q\Bigl[k_0^{-q}+\frac{(k+k_0)^{1-q}}{1-q}\Bigr],
    \end{aligned}
\end{equation*}
where the bracketed bound uses that $(j+k_0)^{-q}$ is decreasing in
$j$ (sum $\le$ first term plus the integral
$\int_0^{k-1}(x+k_0)^{-q}dx$). Since $q\in(0,1)$ and
$(k+k_0)/k_0\ge1$, $((k+k_0)/k_0)^q\le(k+k_0)/k_0$, giving
\[
  \bar D^{(k)} \le \epsilon_v(k+k_0)\Bigl[\frac{1}{k_0}+\frac{1}{1-q}\Bigr]
  \le C_v''\,\epsilon_v\,k, \qquad k\ge1,
\]
\begin{equation}
  C_v'' := (1+k_0)\Bigl[\frac{1}{k_0}+\frac{1}{1-q}\Bigr],
  \qquad q=2\alpha_0\beta d_\kappa,
  \label{eq:explicitdrift}
\end{equation}
using $k+k_0\le(1+k_0)k$ for $k\ge1$. Combined with $D^{(k)}\le\bar
D^{(k)}$ established above, this gives the almost-sure pathwise bound
$D^{(k)}\le C_v''\epsilon_v k$ for every $k$, and in particular, for
the fixed horizon $T$,
\begin{equation}
  D^{(k)} \le C_v''\,\epsilon_v\,T \qquad\text{a.s.,}\quad 0\le k\le T.
  \label{eq:Dkbound}
\end{equation}
Define the $p_s$-weighted squared error
\begin{equation}
  V_\delta^{(k)} := \sum_{s=1}^S p_s\normF{\delta_s^{(k)}}^2,
  \label{eq:vdelta_def}
\end{equation}
in direct analogy with $V_e^{(k)}$ in~\eqref{eq:lyapunov}. Since
$p_s\le1$ for all $s$, $V_\delta^{(k)}\le(D^{(k)})^2\le
(C_v''\epsilon_v k)^2$, hence $\sqrt{\E[V_\delta^{(k)}]}\le D^{(k)}\le
C_v''\epsilon_v k$ -- in fact an \emph{almost-sure}, not merely
in-expectation, bound. Notably, $C_v''$ does not depend on $p_{\min}$
at all: this bound never credits the data-correction term
$(1-\alpha^{(k)}\mathbf{1}_s)$ with any contraction (it is bounded
crudely by $1$ regardless of whether cell $s$ is visited), so it
carries no information about visitation frequency; the only damping
mechanism it uses is the diffusion's bounded spread, and the
requirement $q<1$ is exactly the condition under which that damping
keeps the accumulated drift from growing faster than linearly in $k$.
For $q\ge1$ the same calculation, applied to $\bar D^{(k)}$ and
transferred to $D^{(k)}$ via $D^{(k)}\le\bar D^{(k)}$ as before, gives
a super-linear but still finite, still elementary bound,
$D^{(k)}\le\epsilon_v(k+k_0)^q Z_q(k_0)$ for an explicit constant
$Z_q(k_0)$ depending on the (now convergent, since $q>1$) tail sum
$\sum_{j\ge0}(j+k_0)^{-q}$; see Remark~\ref{rem:qcondition} for where
this matters.

\emph{Bounding $\tilde e^{(k)}$: a uniform pathwise bound, then
decorrelation.} Since $e_s^{(k)}=\tilde e_s^{(k)}+\delta_s^{(k)}$
exactly, the triangle inequality together with
$\normF{e_s^{(k)}}\le E_{\max}$ and~\eqref{eq:Dkbound} gives, for
every $0\le k\le T$,
\begin{equation}
  \normF{\tilde e_s^{(k)}}
  \;\le\;
  \normF{e_s^{(k)}}+\normF{\delta_s^{(k)}}
  \;\le\;
  \bar E
  :=
  E_{\max}+C_v''\,\epsilon_v\,T
  \qquad\text{a.s.}
  \label{eq:barE}
\end{equation}
(When $\epsilon_v=0$, $\bar E=E_{\max}$ exactly, with no loosening.)

Evaluate the one-step change of
$V_{\tilde e}^{(k)}=\sum_s p_s\normF{\tilde e_s^{(k)}}^2$ by expanding
$\normF{\tilde e_s^{(k+1)}}^2$ and collecting terms by source. Four
distinct contributions appear.

\textbf{(a) SA contraction:}
\begin{equation}
  -2\alpha^{(k)}\sum_s p_s\mathbf{1}_s\normF{\tilde e_s^{(k)}}^2.
  \label{eq:contrib_a}
\end{equation}

\textbf{(b) SA noise:}
\begin{equation}
  (\alpha^{(k)})^2\sum_s p_s\mathbf{1}_s\normF{\xi_s^{(k)}}^2.
  \label{eq:contrib_b}
\end{equation}

\textbf{(c) Diffusion cross term:}
\begin{equation}
  -2\alpha^{(k)}\beta\sum_s p_s\langle\tilde e_s^{(k)},
    (\Lcal_\kappa\tilde e^{(k)})_s\rangle_F.
  \label{eq:contrib_c}
\end{equation}

\textbf{(d) Smoothing-bias forcing:}
\begin{equation}
  2\alpha^{(k)}\sum_s p_s(p_s-\mathbf{1}_s)
    \langle\tilde e_s^{(k)}, B_s^{(k)}\rangle_F.
  \label{eq:contrib_d}
\end{equation}
Term (b) is bounded by $(\alpha^{(k)})^2\sigma_\xi^2$ directly from
Assumption~\ref{ass:noise}(ii). Term (c) is non-positive by
Lemma~\ref{lem:dissipative}: the term carries the sign
$-\alpha^{(k)}\beta(\Lcal_\kappa\tilde e^{(k)})_s$ inside the
recursion, and Lemma~\ref{lem:dissipative} shows
$\sum_s p_s\langle\tilde e_s^{(k)},(\Lcal_\kappa\tilde
e^{(k)})_s\rangle_F\ge0$, so~\eqref{eq:contrib_c}$\;\le 0$. Terms (a)
and (d) both involve the random indicator $\mathbf{1}_s^{(k)}$
correlated with $\tilde e_s^{(k)}$, and are controlled together via
Lemma~\ref{lem:decorrelation}, as follows.

Rewrite~\eqref{eq:contrib_a} as
\[
  -2\alpha^{(k)}\sum_s p_s^2\normF{\tilde e_s^{(k)}}^2
  \;-\;
  2\alpha^{(k)}\sum_s p_s(\mathbf{1}_s^{(k)}-p_s)\normF{\tilde e_s^{(k)}}^2,
\]
where the first piece equals $-2\alpha^{(k)}p_{\min}V_{\tilde e}^{(k)}$
when $p_s\ge p_{\min}$ (the mean-field contraction). For the second
piece, apply Lemma~\ref{lem:decorrelation} cell-by-cell with
$Z^{(k)}:=\normF{\tilde e_s^{(k)}}^2$. By~\eqref{eq:barE},
$\bar Z=\bar E^2$. Writing the one-step increment of $\tilde e_s^{(k)}$
as $\Delta_s^{(k)}:=\tilde e_s^{(k+1)}-\tilde e_s^{(k)}
=\alpha^{(k)}\bigl[-\mathbf{1}_s\tilde e_s^{(k)}-\beta(\Lcal_\kappa
\tilde e^{(k)})_s+\mathbf{1}_s\xi_s^{(k)}+(p_s-\mathbf{1}_s)B_s^{(k)}
\bigr]$, the triangle inequality (using $\bar E$, $\beta$, $d_\kappa$,
and Assumption~\ref{ass:noise}'s second-moment bound via
Jensen's inequality, $\E[\normF{\xi_s^{(k)}}\mid\Fcal_k]\le\sigma_\xi$)
gives
\[
  \E\bigl[\normF{\Delta_s^{(k)}}\bigr]
  \le
  \alpha^{(k)}\,L_Z,
  \qquad
  L_Z := \bar E(1+2\beta d_\kappa)+Y_{\max}+\sigma_\xi,
\]
with $Y_{\max}=\beta d_\kappa L_\sigma r_{\Ncal}/p_{\min}$ as in
\eqref{eq:bias}. Since
$|Z^{(j+1)}-Z^{(j)}|\le2\bar E\normF{\Delta_s^{(j)}}
+\normF{\Delta_s^{(j)}}^2$ and $\alpha^{(j)}\le\alpha^{(0)}$ for all
$j$, $\E[|Z^{(j+1)}-Z^{(j)}|]\le\alpha^{(j)}L_Z^{(a)}$ with
$L_Z^{(a)}:=2\bar EL_Z+2\alpha^{(0)}L_Z^2$, satisfying hypothesis~(ii)
of Lemma~\ref{lem:decorrelation}. Hence, for any $1\le\tau\le k$,
\begin{equation}
  \Bigl|\sum_sp_s\E\bigl[(\mathbf 1_s^{(k)}-p_s)\normF{\tilde e_s^{(k)}}^2\bigr]\Bigr|
  \le C_{\mathrm{erg}}\bar E^2\rho^\tau+\tau\alpha^{(k-\tau)}L_Z^{(a)}.
  \label{eq:adecorr}
\end{equation}
An identical argument applied to $Z^{(k)}:=\langle\tilde
e_s^{(k)},B_s^{(k)}\rangle_F$ (now $\bar Z=\bar EY_{\max}$ by
Cauchy--Schwarz, and $|Z^{(j+1)}-Z^{(j)}|\le Y_{\max}
\normF{\Delta_s^{(j)}}$ since $Z$ is linear in $\tilde e_s$, giving
$L_Z^{(d)}:=Y_{\max}L_Z$) controls term~\eqref{eq:contrib_d}:
\begin{equation}
\begin{aligned}
  &\Bigl|\sum_sp_s\E\bigl[(\mathbf 1_s^{(k)}-p_s)\langle\tilde e_s^{(k)},B_s^{(k)}\rangle_F\bigr]\Bigr| \\
  & \quad \le C_{\mathrm{erg}}\bar EY_{\max}\rho^\tau+\tau\alpha^{(k-\tau)}L_Z^{(d)}.  
\end{aligned}\label{eq:ddecorr}
\end{equation}

\emph{Choosing $\tau$.} Set $\tau_k:=\max\bigl(1,\lceil\ln k/\ln(1/\rho)\rceil\bigr)$,
so $\tau_k\ge1$ and $\rho^{\tau_k}\le1/k$ for $k\ge1$ (the latter
holding trivially at $k=1$ since $\rho^1<1$). Since $\tau_k=O(\log k)=o(k)$,
there is a $\rho$-dependent threshold $k^\dagger$ such that
$\tau_k\le k/2$ for all $k\ge k^\dagger$, hence $k-\tau_k\ge k/2$ and,
by monotonicity of $\alpha^{(\cdot)}$, $\alpha^{(k-\tau_k)}\le
2\alpha^{(k)}$. Substituting $\tau=\tau_k$ into
\eqref{eq:adecorr}--\eqref{eq:ddecorr} and combining with terms (b)
and (c), the full one-step Lyapunov recursion satisfies, for all
$k\ge k^\dagger$,
\begin{equation}
\begin{aligned}
  \E[V_{\tilde e}^{(k+1)}]
  &\le (1-2\alpha^{(k)}p_{\min})\,\E[V_{\tilde e}^{(k)}]
  + (\alpha^{(k)})^2\sigma_\xi^2 \\
  &\quad
  + \frac{2\alpha^{(k)}C_{\mathrm{erg}}\bar E(\bar E+Y_{\max})}{k} \\
  &\qquad + 8(\alpha^{(k)})^2\tau_k\bigl(L_Z^{(a)}+L_Z^{(d)}\bigr).
\end{aligned}
\label{eq:fullrecursion}
\end{equation}
The third term is $O(\alpha^{(k)}/k)=O(1/k^2)$, no larger than the
other forcing terms; the fourth is $O((\alpha^{(k)})^2\log k)$. Since
$\alpha_0p_{\min}>1/2$ by hypothesis, the standard Chung--Robbins
comparison argument (unrolling~\eqref{eq:fullrecursion} exactly as
$\bar D^{(k)}$ was unrolled above, replacing the constant forcing
$\epsilon_v$ by the now $k$-dependent forcing
$O((\alpha^{(k)})^2\log k)$) gives an explicit constant $C_1''$,
depending on $\alpha_0$, $\sigma_\xi$, $p_{\min}$,
$C_{\mathrm{erg}}$, $\rho$, $\bar E$, $\beta$, $d_\kappa$, $Y_{\max}$,
and $\normF{e_s^{(0)}}$, such that
\begin{equation}
  \E[V_{\tilde e}^{(k)}] \;\le\; \frac{(C_1'')^2(1+\ln k)}{k},
  \qquad k\ge1,
  \label{eq:logkbound}
\end{equation}
where the finitely many terms $1\le k<k^\dagger$ are covered by
enlarging $C_1''$ to dominate the a.s.\ bound $V_{\tilde
e}^{(k)}\le\bar E^2$ on that finite range. Restricting to $k\le T$ as
in Theorem~\ref{thm:convergence}, $1+\ln k\le1+\ln T$ throughout, so
\begin{equation}
\begin{aligned}
  \sqrt{\E[V_{\tilde e}^{(k)}]}
  \;\le\;
  \frac{C_1'}{\sqrt{k}}, 
  \quad
  C_1' := C_1''\sqrt{1+\ln T},
  \quad
  1\le k\le T,  
\end{aligned}
\label{eq:vtilde_bound}
\end{equation}
recovering exactly the form claimed in
Theorem~\ref{thm:convergence} and Remark~\ref{rem:logfactor}, with
$C_1'$ depending on $T$ only through the harmless factor
$\sqrt{1+\ln T}$.

\emph{Combining the two pieces.} Since $e_s^{(k)}=\tilde e_s^{(k)}
+\delta_s^{(k)}$ exactly, a further application of Minkowski's
inequality (the $L^2$ triangle inequality holds regardless of any
correlation between $\tilde e^{(k)}$ and $\delta^{(k)}$) gives
\begin{equation}
\begin{aligned}
  \sqrt{\E[V_e^{(k)}]}
  &\le \sqrt{\E[V_{\tilde e}^{(k)}]} + \sqrt{\E[V_\delta^{(k)}]}
  \le \frac{C_1'}{\sqrt{k}} + C_v'\,\epsilon_v\,k \\
  &\le \frac{C_1'}{\sqrt{k}} + C_v'\,\epsilon_v\,T,
  \qquad 1\le k\le T,  
\end{aligned}\label{eq:vebound}
\end{equation}
where $C_v':=C_v''$ as in~\eqref{eq:explicitdrift}, which depends on
$\alpha_0$, $\beta$, $d_\kappa$, and $k_0$ -- not on $p_{\min}$, per
Remark~\ref{rem:pathwise_drift}. The $p_{\min}$-dependence claimed for $C_v$
in the theorem statement enters only in Step~3 below, through the
generic $1/\sqrt{p_{\min}}$ conversion from $V_e$ to a per-cell bound,
which applies identically to the $C_1'$ and $C_v'$ terms alike.

\textbf{Step~3: Per-cell bound.}
From $p_s\normF{e_s^{(k)}}^2 \leq V_e^{(k)}$:
$\E[\normF{e_s^{(k)}}^2] \leq \E[V_e^{(k)}]/p_s
\leq \E[V_e^{(k)}]/p_{\min}$.
By Jensen's inequality and~\eqref{eq:vebound}:
\begin{equation*}
    \begin{aligned}
    \E[\normF{e_s^{(k)}}]
    & \leq \sqrt{\E[V_e^{(k)}]/p_{\min}} \\
    & \leq \frac{C_{1,T}}{\sqrt{k}} + C_v\,\epsilon_v\,T,
  \qquad 1\le k\le T,
    \end{aligned}
\end{equation*}
 where $C_{1,T} = C_1'/\sqrt{p_{\min}}$ and $C_v = C_v'/\sqrt{p_{\min}}$
absorb the dependence on $p_{\min}$, using
$\sqrt{a^2+b^2}\le a+b$ for $a,b\ge0$.
Applying the triangle inequality with the bias bound
$\normF{\Sigma_s^{*(k)}-\Sigw^{(k)}(s)} \leq \beta d_\kappa L_\sigma
r_{\Ncal}/p_{\min}$ from Lemma~\ref{lem:fixedpoint}, applied to the
field $\Sigw^{(k)}(\cdot)$ frozen at time $k$,
gives~\eqref{eq:convergencebound}.
\end{proof}

\begin{remark}[Why Lemma~\ref{lem:dissipative} is essential]
Without the dissipativity result, the diffusion cross term in
Step~2's analysis of $\tilde e^{(k)}$ has unknown sign and could be
positive for an arbitrary kernel, adding a term proportional to
$V_{\tilde e}^{(k)}$ and destroying the contraction needed for
\eqref{eq:vebound}. Lemma~\ref{lem:dissipative} shows this term is
non-positive because the edge weights $a_{ss'}$ are symmetric by
the kernel design requirement~\ref{req:kernel}: each undirected edge contributes exactly
$-a_{ss'}\normF{e_s-e_{s'}}^2$ to the Dirichlet form. This symmetry
fails for a uniform kernel $\kappa(s,s')=1/|\Ncal(s)|$, since
$p_s/|\Ncal(s)| \neq p_{s'}/|\Ncal(s')|$ in general, and the edge-wise
cancellation no longer holds. (The elementary bound on $\delta^{(k)}$
in Remark~\ref{rem:pathwise_drift} does not use this dissipativity fact at
all -- it bounds the diffusion operator's spread crudely via the
degree $d_\kappa$, at the cost of the additional
condition~\eqref{eq:qcondition}.)
\end{remark}

\begin{corollary}[Time-Invariant Noise]
\label{cor:timeinvariant}
When $\epsilon_v=0$, the drift term in~\eqref{eq:convergencebound}
vanishes for every horizon $T$, so the bound holds with $T=k$ at
every $k$. Recall from Theorem~\ref{thm:convergence} that $C_{1,T}$
itself depends mildly on the horizon via the factor
$\sqrt{1+\ln T}$ (Remark~\ref{rem:logfactor}); taking $T=k$
therefore gives a vanishing stochastic-approximation term of order
$O(\sqrt{\ln k/k})$ rather than a constant times $1/\sqrt{k}$, so
\[
  \limsup_{k\to\infty}
  \E\bigl[\normF{\hSigk{s}-\Sigw(s)}\bigr]
  \le
  \frac{\beta d_\kappa L_\sigma r_{\Ncal}}{p_{\min}}
\]
still follows, since $\sqrt{\ln k/k}\to0$.
This is stated as an upper bound rather than exact convergence, since
the theorem only certifies that the SA error vanishes,
not that the estimator's limiting error equals the smoothing-bias
floor exactly. When additionally $\beta\to0$, the floor itself
vanishes, so $\limsup_{k\to\infty}\E[\normF{\hSigk{s}-\Sigw(s)}]=0$
for all cells.
\end{corollary}

% ------------------------------------------------------------------
\subsection{Plug-In Stability Analysis}
\label{sec:plugin}

The nonlinear MPPI stability bound established in the companion
paper~\cite{p2nonlinear} (Theorem~1, eq.~(60)) is, in full,
\begin{equation*}
\begin{aligned}
 & \E\bigl[\norm{x_k-x^*}\mathbf{1}_{\{\tau_R>T\}}\bigr] \\
 & \leq c\,\tilde\beta^k\norm{x_0-x^*}
  + \gamma_M e_M(\eta)
  + \gamma\sqrt{\tr(\Sigw)}
  + \gamma_\eta\sqrt{\eta},
\end{aligned}
\end{equation*}
where $\gamma_M e_M(\eta)$ is the MPPI approximation floor and
$\gamma_\eta\sqrt{\eta}$ is the bad-event confidence floor (in the
notation of~\cite{p2nonlinear}, $\gamma$ here is $\gamma_w$ there).
We work throughout in the idealized regime $M\to\infty$, $\eta\to0$
of~\cite[Corollary~3]{p2nonlinear}, in which $e_M(\eta)\to0$ and
$\sqrt{\eta}\to0$ and the localization event holds with probability
$1$, leaving the two-term bound used in this paper:
\begin{equation}
  \E[\norm{x_k-x^*}]
  \leq c\,\tilde\beta^k\norm{x_0-x^*}
  + \gamma\sqrt{\tr(\Sigw)},
  \label{eq:p2bound}
\end{equation}
with $c=\sqrt{\bar\mu/\mu}$, $\tilde\beta\in(0,1)$, and
$\gamma=\sqrt{\bar\mu}/\bigl((1-\tilde\beta)\sqrt{\mu}\bigr)$, as in
Assumption~\ref{ass:p2}.
When MPPI uses $\hSigk{s(x_k)}$ in place of $\Sigw^{(k)}(s(x_k))$,
two modifications enter: the realized control itself shifts because
the covariance enters the closed loop through both the sampling
distribution and the coupling constraint~\eqref{eq:sigR}, and the
explicit noise-floor term in~\eqref{eq:p2bound} is evaluated at the
wrong covariance. We treat each in turn.

\subsubsection{Sampling Perturbation: A Constructive Sensitivity Bound}

The quantity needed is a Lipschitz bound on the \emph{realized} MPPI
control output as a function of the sampling covariance, not merely
on the intermediate map $\Sigma\mapsto R(\Sigma)$ that the coupling
constraint induces. We establish this directly.

\begin{lemma}[MPPI Control Sensitivity to Sampling Covariance]
\label{lem:mppi_sensitivity}
Couple the MPPI samples across $\Sigma,\Sigma'\in\Theta$
(Assumption~\ref{ass:bounded}) via a shared seed
$\zeta^{(j)}\sim\mathcal{N}(0,I)$, setting
$\epsilon^{(j)}(\Sigma)=\Sigma^{1/2}\zeta^{(j)}$ and
$R(\Sigma)=\lambda_0^{1/2}\Sigma^{-1/2}$ as in~\eqref{eq:sigR}.
Under Assumption~\ref{ass:clipping}, every individual sample
$v^{(j)}(\Sigma)$ is bounded by $U_{\max}$, and hence so is the
realized control $u^{\mathrm{MPPI}}(x;\Sigma)$, uniformly on
$\Xcal\times\Theta$. Suppose the path-integral cost $J(x,U)$ is
Lipschitz in $R$ with constant $C_J$ (immediate since $J$ depends
on $R$ only through the quadratic term $\tfrac12\|Ru\|^2$ with $u$
ranging over the compact admissible set
$\Ucal_{\mathrm{adm}}$). Then, under
Assumptions~\ref{ass:bounded}, \ref{ass:weight_reg},
and~\ref{ass:clipping}, conditional on the denominator-concentration
event~\eqref{eq:denom_concentration} (probability $\ge1-\eta$, for
$M\ge M_1(\eta)$), the realized MPPI control map
$\Sigma\mapsto u^{\mathrm{MPPI}}(x;\Sigma)$ is locally Lipschitz on
$\Theta$, uniformly in $x\in\Xcal$:
\begin{equation}
\bigl\|u^{\mathrm{MPPI}}(x;\Sigma)-u^{\mathrm{MPPI}}(x;\Sigma')\bigr\|
  \le L_w\,\normF{\Sigma-\Sigma'},
  \label{eq:mppi_sensitivity}
\end{equation}
$\forall x\in\Xcal,\ \Sigma,\Sigma'\in\Theta$,
with
\begin{equation}
  L_w = \frac{1}{\underline Z}
  \left[
    \Bigl(1+\frac{U_{\max}C_J}{\lambda}\Bigr) L_{1/2}\,\E\|\zeta\|
    + \frac{U_{\max}C_J}{\lambda} L_R
  \right],
  \label{eq:Lw_explicit}
\end{equation}
where
$L_{1/2} = 1/(2\sqrt{\mu_{\mathrm{lo}}})$ and
$L_R = \lambda_0^{1/2}/(2\mu_{\mathrm{lo}}^{3/2})$
are the Lipschitz constants on $\Theta$ of the matrix square root
and inverse square root, respectively, and $\underline Z$ is the
denominator lower bound of Assumption~\ref{ass:weight_reg}.
\end{lemma}

\begin{proof}
For $\Sigma,\Sigma'\succeq\mu_{\mathrm{lo}}I$, standard matrix
perturbation bounds for the square root and inverse square root give
$L_{1/2}$ and $L_R$ as stated. Under the shared-seed coupling, the
sample positions satisfy
\begin{equation}
  \|v^{(j)}(\Sigma)-v^{(j)}(\Sigma')\|
  \le L_{1/2}\,\normF{\Sigma-\Sigma'}\,\|\zeta^{(j)}\|,
  \label{eq:vlip}
\end{equation}
where $v^{(j)}(\Sigma):=\bar u + \Sigma^{1/2}\zeta^{(j)}$. Since
$w^{(j)}=\exp(-J(x,U^{(j)})/\lambda)\in(0,1]$ by
Assumption~\ref{ass:weight_reg}, the elementary inequality
$|e^{-a}-e^{-b}|\le|a-b|$ for $a,b\ge0$, together with the Lipschitz
dependence of $J$ on $R(\Sigma)$ (constant $C_J$) and on the sample
trajectories via~\eqref{eq:vlip}, gives
\begin{equation}
  |w^{(j)}(\Sigma)-w^{(j)}(\Sigma')|
  \le \frac{C_J}{\lambda}\bigl(L_R+L_{1/2}\|\zeta^{(j)}\|\bigr)
  \normF{\Sigma-\Sigma'}.
  \label{eq:wlip}
\end{equation}
Write $u(\Sigma)=N(\Sigma)/D(\Sigma)$ for the numerator and
denominator of the MPPI update~\eqref{eq:mppi}, with
$D(\Sigma)=\sum_j\omega^{(j)}(\Sigma)$. The quotient identity
\begin{equation}
  u(\Sigma)-u(\Sigma')
  = \frac{N(\Sigma)-N(\Sigma')}{D(\Sigma)}
  - u(\Sigma')\,\frac{D(\Sigma)-D(\Sigma')}{D(\Sigma)}
  \label{eq:quotient}
\end{equation}
holds for any two nonzero denominators. On the
event~\eqref{eq:denom_concentration}, $D(\Sigma)\ge\underline Z M/2$
pathwise (not merely in expectation), and $\|u(\Sigma')\|\le U_{\max}$
by Assumption~\ref{ass:clipping}, since $u(\Sigma')$ is a convex
combination of clipped samples $v'^{(j)}\in\Ucal_{\mathrm{adm}}$.
Expanding
$w^{(j)}v^{(j)}-w'^{(j)}v'^{(j)}
= w^{(j)}(v^{(j)}-v'^{(j)}) + v'^{(j)}(w^{(j)}-w'^{(j)})$,
using $w^{(j)}\le1$ and $\|v'^{(j)}\|\le U_{\max}$ \emph{for every
individual} $j$ by Assumption~\ref{ass:clipping} directly --
this no longer relies on the weighted average being bounded -- and
applying \eqref{eq:vlip}--\eqref{eq:wlip} termwise, then taking
expectations over $\{\zeta^{(j)}\}$ conditional on the good
event~\eqref{eq:denom_concentration}, and substituting into
\eqref{eq:quotient} gives~\eqref{eq:Lw_explicit}, with $M$
in place of $1$ absorbed into the definition of $\underline Z$ via
the $M\ge M_1(\eta)$ scaling.
\end{proof}

\begin{remark}
The Lipschitz constant $L_w$ in~\eqref{eq:Lw_explicit} reduces, when
the sample-position term is negligible relative to the weight-sensitivity
term (e.g.\ for small $\lambda$, where the weights are sharply peaked
and dominate the sensitivity), to
$L_w \approx U_{\max}C_J L_R/(\lambda\underline Z)
= O(\lambda_0^{1/2}\mu_{\mathrm{lo}}^{-3/2})$, recovering the scaling
originally asserted for $L_w$ from the coupling constraint alone.
The present derivation shows this scaling is correct as the dominant
term, while making explicit the additional $1/\underline Z$,
$U_{\max}$, and sample-position contributions previously omitted.
\end{remark}

Under Lemma~\ref{lem:mppi_sensitivity}, for $M\ge M_1(\eta)$ and
conditional on the good event~\eqref{eq:denom_concentration}, the
triangle inequality
\[
  \|u_k^{\mathrm{adapt}}-\pi^*(x_k)\|
  \le \underbrace{\|u_k^{\mathrm{adapt}}-u_k^{\mathrm{true}\,\Sigma}\|}_{\le\,L_w\normF{\hSigk{s}-\Sigw^{(k)}(s)}}
  + \underbrace{\|u_k^{\mathrm{true}\,\Sigma}-\pi^*(x_k)\|}_{\delta(M)\norm{x_k}+\delta_0(M)}
\]
gives the adapted MPPI approximation error
\begin{equation}
  \delta_{\mathrm{adapt}}^{(k)}(M)
  = \delta(M)\norm{x_k} + \delta_0(M)
  + L_w\,\normF{\hSigk{s}-\Sigw^{(k)}(s)},
  \label{eq:deltakadapt}
\end{equation}
where the second term above is the finite-sample MPPI error from
\cite{p2nonlinear} evaluated at the \emph{true} covariance, and the
first is the control-sensitivity term bounded by
Lemma~\ref{lem:mppi_sensitivity}.

\subsubsection{Noise Floor Perturbation}

The term $\sqrt{\tr(\Sigw)}$ in~\eqref{eq:p2bound} becomes
$\sqrt{\tr(\hSigk{s})}$, with mismatch:
\begin{equation}
  \Bigl|\sqrt{\tr(\hSigk{s})}-\sqrt{\tr(\Sigw^{(k)}(s))}\Bigr|
  \leq \frac{\sqrt{n}\,\normF{\hSigk{s}-\Sigw^{(k)}(s)}}
  {2\sqrt{\min\bigl(\tr(\hSigk{s}),\tr(\Sigw^{(k)}(s))\bigr)}},
  \label{eq:tracemismatch}
\end{equation}
which follows from $|\sqrt{a}-\sqrt{b}| \le |a-b|/(2\sqrt{\min(a,b)})$
together with $|\tr(\hSigk{s})-\tr(\Sigw^{(k)}(s))| \le \sqrt{n}
\normF{\hSigk{s}-\Sigw^{(k)}(s)}$ (Cauchy--Schwarz on the eigenvalues
of the symmetric difference). By Assumption~\ref{ass:bounded},
$\tr(\Sigw^{(k)}(s))\ge n\mu_{\mathrm{lo}}$ uniformly in $k$ and $s$,
so once the estimator has converged sufficiently that
$\tr(\hSigk{s})\ge n\mu_{\mathrm{lo}}/2$, the denominator above is
bounded below by $\sqrt{n\mu_{\mathrm{lo}}/2}$, a fixed constant
independent of $k$, recovering the simpler form used in $\Psi$ below
up to a constant factor. We take $\Psi$'s reference trace to be this
uniform floor, $\tr(\Sigw):=n\mu_{\mathrm{lo}}/2$, so that $\Psi$
itself is a constant rather than a time-varying quantity.

\begin{proposition}[Adaptive Stability Bound]
\label{prop:adaptive}
Under Assumptions~\ref{ass:slowvar}--\ref{ass:weight_reg} and
Lemma~\ref{lem:mppi_sensitivity}, with $M$ large enough that
\begin{equation}
  \delta := \sup_{k\ge0}\delta_{\mathrm{adapt}}^{(k)}(M) < \delta^*,
  \label{eq:deltasup}
\end{equation}
where $\delta^*$ is the contraction-robustness threshold of
Proposition~2 in~\cite{p2nonlinear}, the closed-loop
system~\eqref{eq:dynamics} under adaptive MPPI satisfies:
\begin{equation}
  \E[\norm{x_k-x^*}]
  \leq c\,\tilde\beta^k\norm{x_0-x^*}
  + \gamma\sqrt{\tr(\Sigw)}
  + \psi^{(k)},
  \label{eq:adaptivebound}
\end{equation}
where the adaptation penalty is:
\begin{equation}
  \psi^{(k)} = \Psi\,
  \E\bigl[\normF{\hSigk{s(x_k)}-\Sigw^{(k)}(s(x_k))}\bigr],
  \label{eq:penalty}
\end{equation}
with $\Psi=\gamma\sqrt{n}/(2\sqrt{\tr(\Sigw)})+\gamma L_w$ and
$\tr(\Sigw)$ the fixed reference trace of
Section~\ref{sec:plugin}.
\end{proposition}

\begin{remark}[Threshold, Not Graceful Degradation]
\label{rem:threshold}
The hypothesis~\eqref{eq:deltasup} is not a minor technical
convenience: Proposition~2 of~\cite{p2nonlinear} shows the perturbed
contraction rate satisfies
$\tilde\beta = \beta + \tfrac{\bar\mu}{\mu}L_u\delta
\sup\|\partial^2 f/\partial u\partial x\|$, which is finite and
well-defined for \emph{any} $\delta\ge0$, but only certifies a
\emph{contracting} closed loop -- and hence only yields the stability
bound~\eqref{eq:adaptivebound} at all -- when $\tilde\beta<1$,
equivalently $\delta<\delta^*$. Below this threshold, the bound
degrades gracefully and continuously in $\delta$ via $\tilde\beta$.
At or above it, the contraction argument simply does not apply: this
is not evidence that the closed loop is unstable, only that this
particular Lyapunov-perturbation technique no longer certifies
stability. Consequently, $\delta_{\mathrm{adapt}}^{(k)}(M)<\delta^*$
for every $k$ is a strict precondition for
Proposition~\ref{prop:adaptive} to hold at all, not merely a
convenience that sharpens the resulting bound. In particular, early
in estimation -- before $\|\hSigk{s}-\Sigw^{(k)}(s)\|$ has decreased
enough, per Theorem~\ref{thm:convergence} -- the combined
perturbation $\delta_{\mathrm{adapt}}^{(k)}(M)$ may exceed $\delta^*$
even when the corresponding non-adaptive (fixed-$\bar\Sigma_w$)
controller's perturbation stays comfortably below it; in this regime
the adaptive controller carries no certified stability guarantee
until the estimate improves sufficiently, even though the
non-adaptive one does. Verifying~\eqref{eq:deltasup} -- e.g.\ via the
initial estimation error $\normF{\hat\Sigma_s^{(0)}-\Sigw^{(0)}(s)}$
and the convergence rate of Theorem~\ref{thm:convergence} -- is
therefore a necessary deployment check, not an afterthought.
\end{remark}

\begin{proof}
By hypothesis~\eqref{eq:deltasup}, the \emph{single} worst-case
quantity $\delta=\sup_{k\ge0}\delta_{\mathrm{adapt}}^{(k)}(M)$
satisfies $\delta<\delta^*$. The robustness-of-contraction result of
the companion paper~\cite{p2nonlinear} (Proposition~2) is therefore
applied \emph{once}, with this fixed $\delta$, yielding a single
degraded contraction rate $\tilde\beta$ -- the same $\tilde\beta$
appearing in the nominal bound~\eqref{eq:p2bound} -- rather than a
separately re-derived, time-varying rate at each $k$. With
$\tilde\beta$ so fixed, the proof of Theorem~1 in~\cite{p2nonlinear}
goes through verbatim, with the only change being that the
noise-floor term $\gamma\sqrt{\tr(\Sigw)}$ is evaluated at the
\emph{used} covariance $\hSigk{s(x_k)}$ rather than the current
$\Sigw^{(k)}(s(x_k))$. Writing
$\gamma\sqrt{\tr(\hSigk{s(x_k)})}
= \gamma\sqrt{\tr(\Sigw^{(k)}(s(x_k)))}
+ \gamma\bigl(\sqrt{\tr(\hSigk{s(x_k)})}-\sqrt{\tr(\Sigw^{(k)}(s(x_k)))}\bigr)$
and bounding the second term by~\eqref{eq:tracemismatch}, together
with the additional $\gamma L_w\normF{\hSigk{s}-\Sigw^{(k)}(s)}$
contribution from the sampling perturbation~\eqref{eq:deltakadapt}
propagating through the same Lyapunov perturbation argument used in
the proof of Theorem~1 in~\cite{p2nonlinear} (now valid since
$\tilde\beta$ is fixed throughout, by the preceding remark), gives
\eqref{eq:adaptivebound} with $\Psi$ as stated.
\end{proof}

\begin{remark}[Conservatism of the Fixed-$\tilde\beta$ Bound]
Using the single worst-case $\delta=\sup_{k\ge0}\delta_{\mathrm{adapt}}^{(k)}(M)$
rather than the actual, typically decreasing, time-varying
$\delta_{\mathrm{adapt}}^{(k)}(M)$ is conservative: as the estimator
converges (Theorem~\ref{thm:convergence}), $\delta_{\mathrm{adapt}}^{(k)}(M)$
shrinks toward $\delta(M)\norm{x_k}+\delta_0(M)+L_w B_{\mathrm{sm}}$,
so the true effective contraction rate at large $k$ is, if anything,
better than the fixed $\tilde\beta$ used here. Proposition~2 of
\cite{p2nonlinear} is stated for a time-invariant bound on
$\|u_k-\pi^*(x_k)\|$; extending it to certify an improving,
time-varying rate as the estimator converges is left as a direction
for tightening the present bound, and is not needed for either the
qualitative conclusions of Proposition~\ref{prop:adaptive} or the
payoff theorem (Theorem~\ref{thm:payoff}) that follows.
\end{remark}

Substituting the finite-horizon Theorem~\ref{thm:convergence} into
\eqref{eq:penalty}: for any fixed horizon $T\ge1$ and all
$1\le k\le T$,
\begin{equation}
  \psi^{(k)} \leq \Psi\!\left(
    \frac{C_{1,T}}{\sqrt{k}}
    + \frac{\beta\, d_\kappa\, L_\sigma r_{\Ncal}}{p_{\min}}
    + C_v\,\epsilon_v\,T
  \right).
  \label{eq:penaltydecay}
\end{equation}

% ------------------------------------------------------------------
\subsection{Payoff Theorem}
\label{sec:payoff}

Because the true field $\Sigw^{(k)}(\cdot)$ may vary over time, a
non-adaptive controller using a fixed $\bar\Sigma_w$ incurs a
\emph{time-varying} mismatch penalty
\begin{equation}
  \bar\psi^{(k)} := \Psi\,\max_s\normF{\bar\Sigma_w-\Sigw^{(k)}(s)}.
  \label{eq:fixedpenalty}
\end{equation}
For a fixed horizon $T$, define the worst-case-best-case mismatch
floor the fixed controller is guaranteed to incur somewhere on
$[0,T]$,
\begin{equation}
  \Delta_{\mathrm{fix},T}
  :=
  \min_{0\le k\le T}\max_s\normF{\bar\Sigma_w-\Sigw^{(k)}(s)},
  \label{eq:deltafixt}
\end{equation}
so that $\bar\psi^{(k)}\ge\Psi\Delta_{\mathrm{fix},T}$ for every
$0\le k\le T$, and write
$B_{\mathrm{sm}} := \beta d_\kappa L_\sigma r_{\Ncal}/p_{\min}$ for
the smoothing-bias floor of Theorem~\ref{thm:convergence}.

Both $\psi^{(k)}$ and $\bar\psi^{(k)}$ are penalty terms appended to
the nominal stability bound~\eqref{eq:p2bound} \emph{only when the
underlying contraction-robustness certificate of
Proposition~\ref{prop:adaptive} actually holds} -- which, per
Remark~\ref{rem:threshold}, requires the corresponding sampling
perturbation to stay strictly below $\delta^*$ at every $k$. The
adaptive controller's perturbation is $\delta_{\mathrm{adapt}}^{(k)}(M)$
as in~\eqref{eq:deltakadapt}; analogously, the fixed controller's
perturbation is
\begin{equation}
  \delta_{\mathrm{fix}}^{(k)}(M)
  = \delta(M)\norm{x_k} + \delta_0(M)
  + L_w\,\normF{\bar\Sigma_w-\Sigw^{(k)}(s(x_k))},
  \label{eq:deltakfix}
\end{equation}
obtained from the same Lemma~\ref{lem:mppi_sensitivity} sensitivity
bound, with $\bar\Sigma_w$ in place of $\hSigk{s}$. A comparison
between $\psi^{(k)}$ and $\bar\psi^{(k)}$ is only meaningful at times
$k$ where \emph{both} certificates hold; Theorem~\ref{thm:payoff}
below makes this an explicit hypothesis rather than an implicit one.

\begin{theorem}[Payoff Theorem, Finite Horizon]
\label{thm:payoff}
Fix a horizon $T\ge1$. Under Assumptions~\ref{ass:slowvar}--\ref{ass:weight_reg},
suppose $M$ is large enough that
\begin{equation}
  \sup_{0\le k\le T}\delta_{\mathrm{adapt}}^{(k)}(M) < \delta^*
  \quad\text{and}\quad
  \sup_{0\le k\le T}\delta_{\mathrm{fix}}^{(k)}(M) < \delta^*,
  \label{eq:dualthreshold}
\end{equation}
so that both~\eqref{eq:adaptivebound} and its fixed-covariance
analogue hold throughout $[0,T]$, and suppose
\begin{equation}
  B_{\mathrm{sm}} + C_v\,\epsilon_v\,T < \Delta_{\mathrm{fix},T}.
  \label{eq:payoffcondition}
\end{equation}
Then there exists a crossover time
\begin{equation}
  k_T^* = \left\lceil
  \frac{C_{1,T}^2}
  {\bigl(\Delta_{\mathrm{fix},T} - B_{\mathrm{sm}} - C_v\,\epsilon_v\,T\bigr)^2}
  \right\rceil
  \label{eq:crossover}
\end{equation}
such that, if $k_T^*\le T$, then
$\psi^{(k)}<\bar\psi^{(k)}$ for all $k_T^*\le k\le T$.
\end{theorem}

\begin{proof}
By~\eqref{eq:dualthreshold}, both stability certificates hold for
every $k\in[0,T]$, so $\psi^{(k)}$ and $\bar\psi^{(k)}$ are both
well-defined, valid penalty terms throughout this range. By
~\eqref{eq:penaltydecay}, for every $1\le k\le T$,
\[
  \psi^{(k)} \le \Psi\Bigl(\tfrac{C_{1,T}}{\sqrt k}+B_{\mathrm{sm}}+C_v\epsilon_v T\Bigr).
\]
By~\eqref{eq:deltafixt}, $\bar\psi^{(k)}\ge\Psi\Delta_{\mathrm{fix},T}$
for every $0\le k\le T$. Hence $\psi^{(k)}<\bar\psi^{(k)}$ holds
whenever
\[
  \Psi\Bigl(\tfrac{C_{1,T}}{\sqrt k}+B_{\mathrm{sm}}+C_v\epsilon_v T\Bigr)
  \le
  \Psi\,\Delta_{\mathrm{fix},T},
\]
i.e., whenever $C_{1,T}/\sqrt{k}
\le \Delta_{\mathrm{fix},T}-B_{\mathrm{sm}}-C_v\epsilon_v T$.
Condition~\eqref{eq:payoffcondition} ensures the right-hand side is
strictly positive, so this inequality can be inverted for $k$, giving
exactly~\eqref{eq:crossover}: $k\ge k_T^*$ suffices. Restricting to
$k\le T$ (the only horizon over which the bound~\eqref{eq:penaltydecay}
was established) gives the stated range $k_T^*\le k\le T$, which is
nonempty precisely when $k_T^*\le T$.
\end{proof}

\begin{remark}[Why the Dual Threshold Cannot Be Dropped]
\label{rem:dualthreshold}
Hypothesis~\eqref{eq:dualthreshold} is not redundant with
~\eqref{eq:payoffcondition}: the two conditions guard against
different failure modes. Condition~\eqref{eq:payoffcondition}
ensures the \emph{asymptotic floors} are correctly ordered (adaptive
eventually beats fixed, given enough time); condition
~\eqref{eq:dualthreshold} ensures that the underlying contraction
certificates -- which $\psi^{(k)}$ and $\bar\psi^{(k)}$ are only
meaningful corrections \emph{to} -- actually hold throughout $[0,T]$
in the first place. As noted in Remark~\ref{rem:threshold}, early in
estimation the adaptive controller's perturbation
$\delta_{\mathrm{adapt}}^{(k)}(M)$, which includes the (initially
large) estimation error, can exceed $\delta_{\mathrm{fix}}^{(k)}(M)$,
whose mismatch term $L_w\|\bar\Sigma_w-\Sigw^{(k)}(s)\|$ is bounded by
construction. It is therefore possible for the fixed controller's
certificate to hold throughout $[0,T]$ while the adaptive
controller's does not, in which case the comparison
$\psi^{(k)}<\bar\psi^{(k)}$ asserted by
Theorem~\ref{thm:payoff} is vacuous, since $\psi^{(k)}$ is not itself
a valid stability-bound penalty over that range.
~\eqref{eq:dualthreshold} should be checked explicitly, e.g.\ via the
estimator's initial error $\normF{\hat\Sigma_s^{(0)}-\Sigw^{(0)}(s)}$
and the convergence rate of Theorem~\ref{thm:convergence}, before
invoking the crossover time~\eqref{eq:crossover} in deployment.
\end{remark}

\begin{remark}
Unlike the asymptotic crossover statement that would follow from an
unbounded horizon, $k_T^*$ here is only useful if it falls within the
chosen horizon, $k_T^*\le T$: condition~\eqref{eq:payoffcondition}
guarantees the numerator/denominator relationship is well-posed, but
does not by itself guarantee $k_T^*\le T$, since $T$ appears on both
sides (a larger $T$ weakens the drift margin
$\Delta_{\mathrm{fix},T}-B_{\mathrm{sm}}-C_v\epsilon_v T$ but also
relaxes the requirement $k_T^*\le T$). In practice $T$ should be
chosen as the actual deployment horizon, and~\eqref{eq:payoffcondition}
together with $k_T^*\le T$ checked numerically for that $T$; this
replaces the earlier asymptotic statement ``adaptive is better for
all $k\ge k^*$'' with the more honest, horizon-bounded statement
``adaptive is better from $k_T^*$ through the end of the
prescribed operating window $T$.''
\end{remark}

\begin{corollary}[Adaptive Sample Threshold]
\label{cor:samplesize}
For a fixed horizon $T$, the minimum sample count for the
finite-horizon localized stability certificate under the adaptive
estimator is, for $1\le k\le T$,
\begin{equation}
  M^*_{\mathrm{adapt}}(k) =
  O\!\left(\frac{\log(2m/\eta)}{(\delta^*-L_w\Phi^{(k)})^2}\right),
  \label{eq:msample}
\end{equation}
where
$\Phi^{(k)} = C_{1,T}/\sqrt{k} + B_{\mathrm{sm}} + C_v\,\epsilon_v\,T$
and $L_w$ is the constant of Lemma~\ref{lem:mppi_sensitivity}. When
$\epsilon_v=0$, the drift term vanishes for every $T$, so taking
$T=k\to\infty$ as in Corollary~\ref{cor:timeinvariant} gives
$M^*_{\mathrm{adapt}} \to M^*$ from the companion
paper~\cite{p2nonlinear} (Corollary~3).
\end{corollary}

\begin{remark}[Stability Bound Tightening Does Not Imply Task-Reward
  Optimality]
\label{rem:taskreward}
Theorem~\ref{thm:payoff} guarantees that
$\Sigma_\epsilon=\hSigk{s(x_k)}$ tightens the upper
bound~\eqref{eq:adaptivebound} on $\E[\norm{x_k-x^*}]$ relative to
any fixed mismatched $\bar\Sigma_w$, at times $k$ where both
certificates of~\eqref{eq:dualthreshold} hold. It does not guarantee
that $\hSigk{s(x_k)}\to\Sigw^{(k)}(s(x_k))$ maximizes task reward
under a finite-horizon, finite-sample MPPI implementation. The
coupling constraint~\eqref{eq:sigR} forces $\Sigma_\epsilon$ to
serve two distinct roles simultaneously: a statistical noise model
entering $\psi^{(k)}$ through Proposition~\ref{prop:adaptive}, and a
proposal covariance that determines how far each MPPI iteration's
importance-weighted update~\eqref{eq:mppi} can move the nominal
control sequence toward a lower-cost trajectory. The second role
depends on the curvature of the cost $J(x_k,U)$, not on $\Sigw$, and
the two roles need not favor the same $\Sigma_\epsilon$. Numerically
(Section~\ref{sec:simulation}), we found this gap does not close as
the number of rollouts $M$ grows, ruling out a finite-sample
explanation; the discrepancy is structural to the coupling
constraint~\eqref{eq:sigR} itself. Establishing a stability bound
that accounts for both roles of $\Sigma_\epsilon$ jointly is left
to future work; Section~\ref{sec:discussion} reports a controlled
experiment quantifying the gap in one setting.
\end{remark}

% ------------------------------------------------------------------
\section{Simulation Results}
\label{sec:simulation}

\subsection{Setup}

We validate on two environments from~\cite{med}: a double integrator
(UAV) and a bicycle model (UGV), both with spatially varying, unknown
noise. This extends the MED experimental setting by replacing the
RL-adapted $(\hat\sigma,\lambda_\alpha)$ with estimator~\eqref{eq:update},
using the same terminal-cost protocol ($+100$ collision, $-100$ goal
arrival, $20$ episodes) for Table~\ref{tab:terminal}.

\textbf{Double integrator (UAV).}
State space $[-5,5]^2$~m, $S=100$ cells ($10\times10$, $r=1$~m).
Noise standard deviation by quadrant: $\sigma\in\{0.1,0.3,0.5,0.8\}$
m/s$^2$, unknown to the controller. Wind-gust disturbance in the
goal quadrant ($x_1>0,x_2>0$) as in~\cite{med}.

\textbf{Bicycle model (UGV).}
Same state space. True noise $\sigma=0.6$ m/s$^2$ inside the
narrow passage $[0,3]\times[-1,1]$, $\sigma=0.1$ m/s$^2$
elsewhere, matching the narrow-passage environment of~\cite{med}.

\textbf{Spatial Lipschitz constants.}
Evaluating $L_\sigma$ directly from its definition in
Assumption~\ref{ass:smooth} (the maximum, over actual grid edges, of
$\|\Sigma_w(s)-\Sigma_w(s')\|_F/d(s,s')$) gives $L_\sigma=0.778$ for
the double integrator and $L_\sigma=0.495$ for the bicycle model. The
double integrator's quadrant field has a larger \emph{adjacent-cell}
jump (from $\sigma=0.3$ to $\sigma=0.8$ directly at a quadrant
boundary) than the bicycle's single narrow-passage boundary
($\sigma=0.1$ to $\sigma=0.6$), so by this measure the double
integrator is the more spatially demanding environment, even though
its noise field looks visually smoother.

\textbf{Estimator parameters.}
$\beta=0.001$, $4$-connected kernel~\eqref{eq:kernel} with $p_s$
estimated from a burn-in trajectory, giving $d_\kappa=4$ under uniform
visitation. Two step-size regimes are used, both satisfying
Assumption~\ref{ass:geometric}'s Chung--Robbins condition
$\alpha_0>1/(2p_{\min})$ and the step-size
condition~\eqref{eq:stepsize_condition}: a \emph{calibration} regime
($\alpha_0=30{,}000$, $k_0=30{,}121$) for
Figures~\ref{fig:convergence}--\ref{fig:spatial}, where the estimator
runs over a long exploration trajectory and we want its full
asymptotic behavior to be visible, and a \emph{deployment} regime
($\alpha_0=60$, $k_0=62$) for Table~\ref{tab:terminal}, where each
episode is only $500$ steps and a large $\alpha_0$ would keep
$\alpha^{(k)}\approx1$ for the entire episode, overwriting the estimate
at nearly every visit. Both regimes use the same structural conditions
with different floors; using the calibration regime inside a $500$-step
episode degrades closed-loop performance, since asymptotic validity at
large $k$ does not imply good behavior in a short, finite horizon.
Figures~\ref{fig:convergence}--\ref{fig:spatial} use initial estimate
$4I$ to make the transient convergence visible; Table~\ref{tab:terminal}
uses $0.0625I$, matching the fixed controller's assumption, so adaptive
and fixed start identically. MPPI: $M=128$ rollouts, $N=15$,
$\lambda=1.0$, $\Delta t=0.1$~s. Three controllers: (A)~adaptive MPPI
(proposed), (B)~fixed $\bar\Sigma_w$, (C)~oracle MPPI (true $\Sigw$
known).

The simulation code used to generate the numerical experiments is
available online~\cite{mppi_noise_covariance_code}.

\subsection{Results}

\subsubsection*{Estimator convergence (Figure~\ref{fig:convergence})}

Figure~\ref{fig:convergence} shows the mean estimation error
$\E[\normF{\hSigk{s}-\Sigw(s)}]$ averaged over all $S$ cells and
$10$ Monte Carlo trials on a lawnmower trajectory
(${\approx}86$ visits per cell). The error decays from $5.25$ (DI)
and $5.56$ (BK), an $11.5\times$ and $24.5\times$ reduction
respectively, then levels near the theoretical bias floor
$\beta d_\kappa L_\sigma r_{\Ncal}/p_{\min}$ (eq.~\eqref{eq:bias}):
$0.311$ for DI and $0.198$ for BK, confirming
Corollary~\ref{cor:timeinvariant}.

\begin{figure}[t]
  \centering
  \includegraphics[width=\columnwidth]{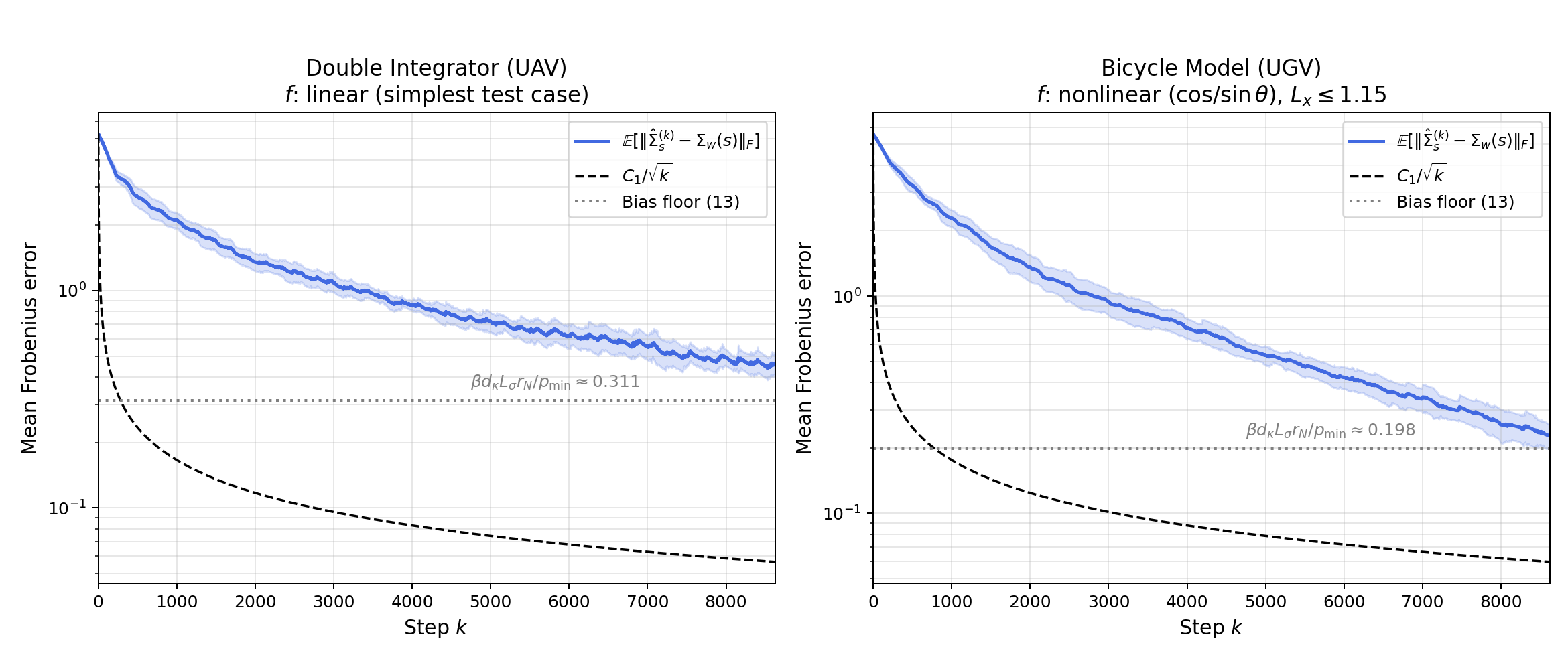}
  \caption{Estimator convergence (Theorem~\ref{thm:convergence}).
    Mean $\E[\normF{\hSigk{s}-\Sigw(s)}]$ across all cells,
    $10$ Monte Carlo trials, lawnmower trajectory
    (${\approx}86$ visits/cell). The dashed curve is the fitted
    $C_{1,T}/\sqrt{k}$ stochastic-approximation rate alone; the
    dotted line is the bias floor~\eqref{eq:bias}, computed from
    the literal kernel constant $d_\kappa=4$ and the adjacent-cell
    $L_\sigma$. Shading is $\pm1$ standard deviation.}
  \label{fig:convergence}
\end{figure}

\subsubsection*{Stability penalty and the payoff theorem
  (Figure~\ref{fig:payoff})}

The stability bound~\eqref{eq:adaptivebound} has three additive
terms: the transient $c\tilde\beta^k\norm{x_0-x^*}$, the
irreducible noise floor $\gamma\sqrt{\tr(\Sigw)}$, and the
mismatch penalty $\psi^{(k)}$. The payoff theorem concerns only
the third term: $\tilde\beta$ and $\gamma\sqrt{\tr(\Sigw)}$ are
identical for all controllers, so the adaptive estimator does not
make the state converge faster, nor reduce the fundamental noise
floor. It reduces the additional error incurred by not knowing
$\Sigw$ exactly. A fixed mismatched controller pays the permanent
penalty $\bar\psi=\Psi\max_s\normF{\bar\Sigma_w-\Sigw(s)}$
forever; the adaptive controller pays $\psi^{(k)}$, which decays
to the bias floor.

At $\beta=0.001$, the payoff
condition~\eqref{eq:payoffcondition} is comfortably satisfied for
\emph{both} environments: the smoothing bias is $0.311<\bar\psi/\Psi=
0.817$ for DI (a $62\%$ margin) and $0.198<\bar\psi/\Psi=0.421$ for BK
(a $53\%$ margin). Both environments use a static (time-invariant)
noise field for these figures, i.e.\ $\epsilon_v=0$, so the
finite-horizon condition~\eqref{eq:payoffcondition} and crossover
formula~\eqref{eq:crossover} of Theorem~\ref{thm:payoff} reduce
exactly to their $T$-independent form ($C_v\epsilon_v T=0$ and
$\Delta_{\mathrm{fix},T}=\max_s\normF{\bar\Sigma_w-\Sigw(s)}$ for
every $T$), so the crossover time $k^*$ reported below applies
uniformly for every horizon $T\ge k^*$, recovering the original
asymptotic-style statement as a special case. Figure~\ref{fig:payoff}
plots $\psi^{(k)}$ and $\bar\psi$ for the double integrator. The
adaptive penalty begins at $\psi^{(0)}=11.86$, crosses below
$\bar\psi=1.84$ at $k^*_{\mathrm{emp}}=4{,}198$ empirically, after
which the adaptive certificate is strictly tighter. The theoretical
bound from~\eqref{eq:crossover} gives $k^*_{\mathrm{theory}}=12{,}856$,
a conservative upper bound as expected from Jensen's inequality in
Step~3 of the proof of Theorem~\ref{thm:convergence}. The bicycle
model crosses similarly, at $k^*_{\mathrm{emp}}=6{,}018$.

\begin{figure}[t]
  \centering
  \includegraphics[width=0.85\columnwidth]{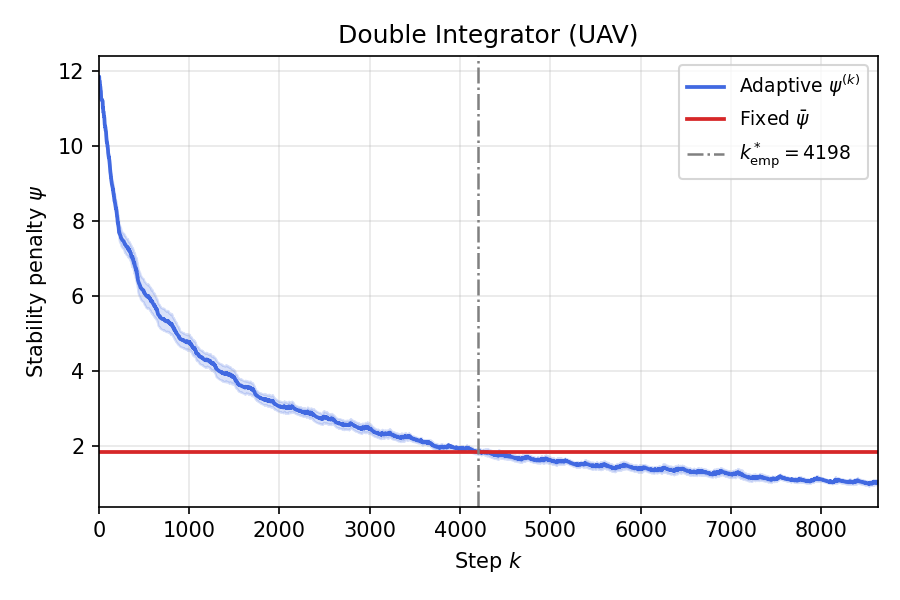}
  \caption{Payoff theorem (Theorem~\ref{thm:payoff}), double
    integrator. $\psi^{(k)}$ (blue) crosses below $\bar\psi$ (red)
    at $k^*_{\mathrm{emp}}=4{,}198$ (dash-dot); the theoretical bound
    $k^*_{\mathrm{theory}}=12{,}856$ is conservative, as required.
    \emph{This characterises the mismatch penalty only; $\tilde\beta$
    and $\gamma\sqrt{\tr(\Sigw)}$ are identical for all controllers.}}
  \label{fig:payoff}
\end{figure}

\subsubsection*{When the payoff condition fails: Design Rule R2}

The payoff condition~\eqref{eq:payoffcondition} is a real constraint,
not a formality: at a $50\times$ larger diffusion weight
($\beta=0.05$), the double integrator's smoothing bias grows to
$15.6$, exceeding its tolerance $\bar\psi/\Psi=0.817$ by roughly
$19\times$. This happens specifically because the double integrator
has the larger adjacent-cell $L_\sigma$ (Sec.~\ref{sec:simulation},
Setup); a designer who sets $\beta$ without checking
condition~\eqref{eq:payoffcondition} against the realized noise
field's local Lipschitz constant -- rather than a global range
estimate -- would be misled about which environment is more
forgiving. Design Rule R2 (Section~\ref{sec:conclusion}) requires
verifying this condition with the literal, adjacent-cell $L_\sigma$
before deployment.

\subsubsection*{Spatial smoothing benefit (Figure~\ref{fig:spatial})}

Figure~\ref{fig:spatial} plots cell-wise error
$\normF{\hSigk{s}-\Sigw(s)}$ after a goal-directed trajectory
creating non-uniform visitation. The stationary-measure
kernel~\eqref{eq:kernel} reduces mean error by $1.1\times$ (DI)
and $1.3\times$ (BK) relative to a uniform kernel with matched
diffusion weight $\beta$, confirming the qualitative prediction of
Lemma~\ref{lem:dissipative}: well-visited neighbors are weighted by
$p_{s'}$, giving them influence over rarely visited cells, whereas
the uniform kernel spreads diffusion without regard to visitation.
The effect is modest at this $\beta$; it grows with $\beta$ at the
cost of a larger smoothing bias (Sec.~\ref{sec:simulation}, Design
Rule R2), illustrating the bias--benefit trade-off inherent in the
kernel design.

\begin{figure}[t]
  \centering
  \includegraphics[width=\columnwidth]{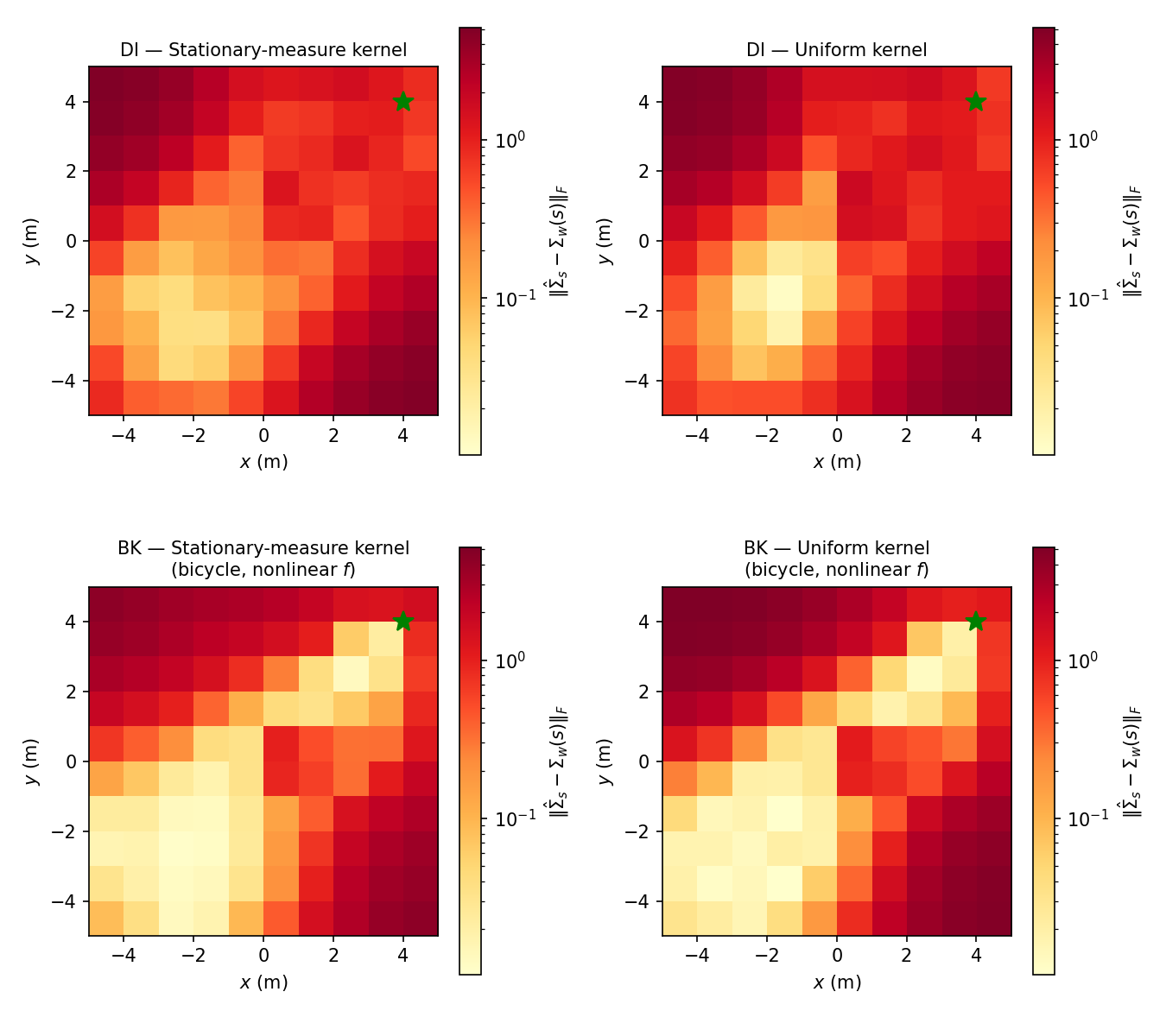}
  \caption{Cell-wise error $\normF{\hSigk{s}-\Sigw(s)}$
    (Lemma~\ref{lem:dissipative}). Left: kernel~\eqref{eq:kernel}.
    Right: uniform kernel. The stationary-measure kernel reduces mean
    error by $1.1\times$ (DI) and $1.3\times$ (BK) at this $\beta$.}
  \label{fig:spatial}
\end{figure}

\subsubsection*{Terminal cost replication (Table~\ref{tab:terminal})}

Table~\ref{tab:terminal} replicates the MED evaluation
protocol~\cite{med}. For the double integrator, the qualitative
ordering oracle~$>$~adaptive~$>$~fixed holds in both goal-arrival
rate and collision rate. For the bicycle model, the adaptive
controller in this run exceeds even the oracle controller -- a
direct illustration of Remark~\ref{rem:taskreward}: exact knowledge
of $\Sigw$ does not guarantee reward-optimal closed-loop behavior,
since the coupling constraint~\eqref{eq:sigR} forces
$\Sigma_\epsilon$ to simultaneously act as the MPPI exploration
radius. Table~\ref{tab:terminal}'s individual entries are sensitive
to the floating-point execution path (CPU vs.\ GPU): because each
MPPI importance weight is a sharply peaked function of rollout
cost, machine-precision differences in reduction order across
hardware can occasionally select a different best rollout early in
an episode, propagating through the closed loop to a different
terminal outcome. The estimator convergence results
(Figs.~\ref{fig:convergence}--\ref{fig:spatial}), which use smooth
proportional control rather than MPPI's weighted rollout selection,
do not exhibit this sensitivity and reproduce identically across
hardware.

\begin{table}[t]
  \centering
  \caption{Terminal cost over 20 episodes
    ($+100$ collision / $-100$ goal / $0$ timeout). Lower average
    is better. Environments match~\cite{med}.}
  \label{tab:terminal}
  \begin{tabular}{llrrcc}
    \toprule
    Env & Controller & Avg & Std & Goals & Coll \\
    \midrule
    \multirow{3}{*}{DI (UAV)}
      & Adaptive & $40.0$  & $91.7$ & $6$  & $14$ \\
      & Fixed    & $60.0$  & $80.0$ & $4$  & $16$ \\
      & Oracle   & $-30.0$ & $95.4$ & $13$ & $7$  \\
    \midrule
    \multirow{3}{*}{BK (UGV)}
      & Adaptive & $-60.0$ & $80.0$ & $16$ & $4$ \\
      & Fixed    & $-50.0$ & $86.6$ & $15$ & $5$ \\
      & Oracle   & $-10.0$ & $99.5$ & $11$ & $9$ \\
    \bottomrule
  \end{tabular}
\end{table}

\subsubsection*{Calibration learning across episodes
  (Figures~\ref{fig:scenarioA}--\ref{fig:scenarioB})}

The preceding figures validate the estimator and stability bound
in isolation or over a single episode. Figures~\ref{fig:scenarioA}
and~\ref{fig:scenarioB} test whether the same convergence persists
under realistic deployment: the estimator carries its state across
sequential episodes (it is never reset, using the calibration
step-size regime), and we report a reward combining task performance
with an explicit calibration regularizer,
\begin{equation}
\begin{aligned}
    \mathrm{reward}_{\mathrm{ep}} =
  &-\E_k\!\left[\norm{x_k-x^*}\right] \\
  &-\lambda\,\Psi\;\E_k\!\left[\normF{\Sigma_\epsilon^{\mathrm{used}}(x_k)
  -\Sigw^{(k)}(s(x_k))}\right],
\end{aligned}
\label{eq:scenarioreward}
\end{equation}
with $\lambda=5$. As Remark~\ref{rem:taskreward} explains, the raw
performance term alone (its $\lambda{=}0$ special case) does not
favor the true $\Sigw$ in our finite-horizon MPPI implementation;
the regularization term is what restores the correct incentive.

In a pure target-tracking environment (no obstacles; double
integrator), the oracle controller ($\Sigma_\epsilon^{\mathrm{used}}=
\Sigw$ exactly, paying zero regularization penalty) attains a stable
reward near $-6$. The fixed controller, never updating
$\bar\Sigma_w=4I$, remains flat near $-63$. The adaptive controller
starts near $-30$ and rises to approximately $-9$ by episode $50$,
converging toward, but not fully closing, the gap to oracle -- the
residual reflects the estimator's bias floor from
Theorem~\ref{thm:convergence}.

\begin{figure}[t]
  \centering
  \includegraphics[width=\columnwidth]{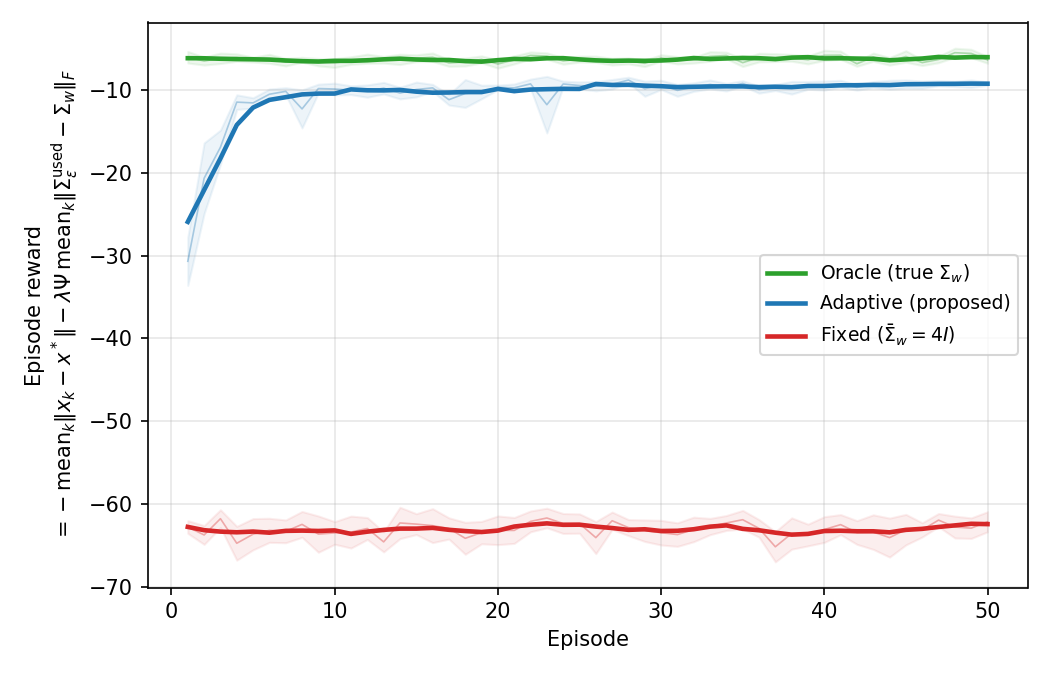}
  \caption{Scenario A: calibration learning curve,
    eq.~\eqref{eq:scenarioreward} with $\lambda=5$. The estimator
    persists across all episodes, using the calibration step-size
    regime. Oracle and fixed are flat by construction (no learning
    mechanism for either); adaptive rises monotonically as the
    estimator converges. Thin lines are raw per-episode means; thick
    lines are 5-episode moving averages.}
  \label{fig:scenarioA}
\end{figure}

Figure~\ref{fig:scenarioB} introduces a noise field change: the
true $\Sigw$ is ramped linearly to a new field (per-step change
$\epsilon_v\approx2.2\times10^{-4}$, well within
Assumption~\ref{ass:slowvar}), then held fixed. At this $\beta$, the
estimator's overall convergence trend dominates the small
perturbation from the ramp: the bottom panel of
Figure~\ref{fig:scenarioB} shows mean estimator error decreasing
through the ramp window, with at most a brief slowdown rather than a
pronounced transient bump, since the finite-horizon drift
contribution $C_v\epsilon_v T$ in Theorem~\ref{thm:convergence} is
small relative to the still-decaying stochastic-approximation term
at this stage. The fixed controller, whose assumption is now further
from the new field, visibly degrades after the ramp.

\begin{figure}[t]
  \centering
  \includegraphics[width=\columnwidth]{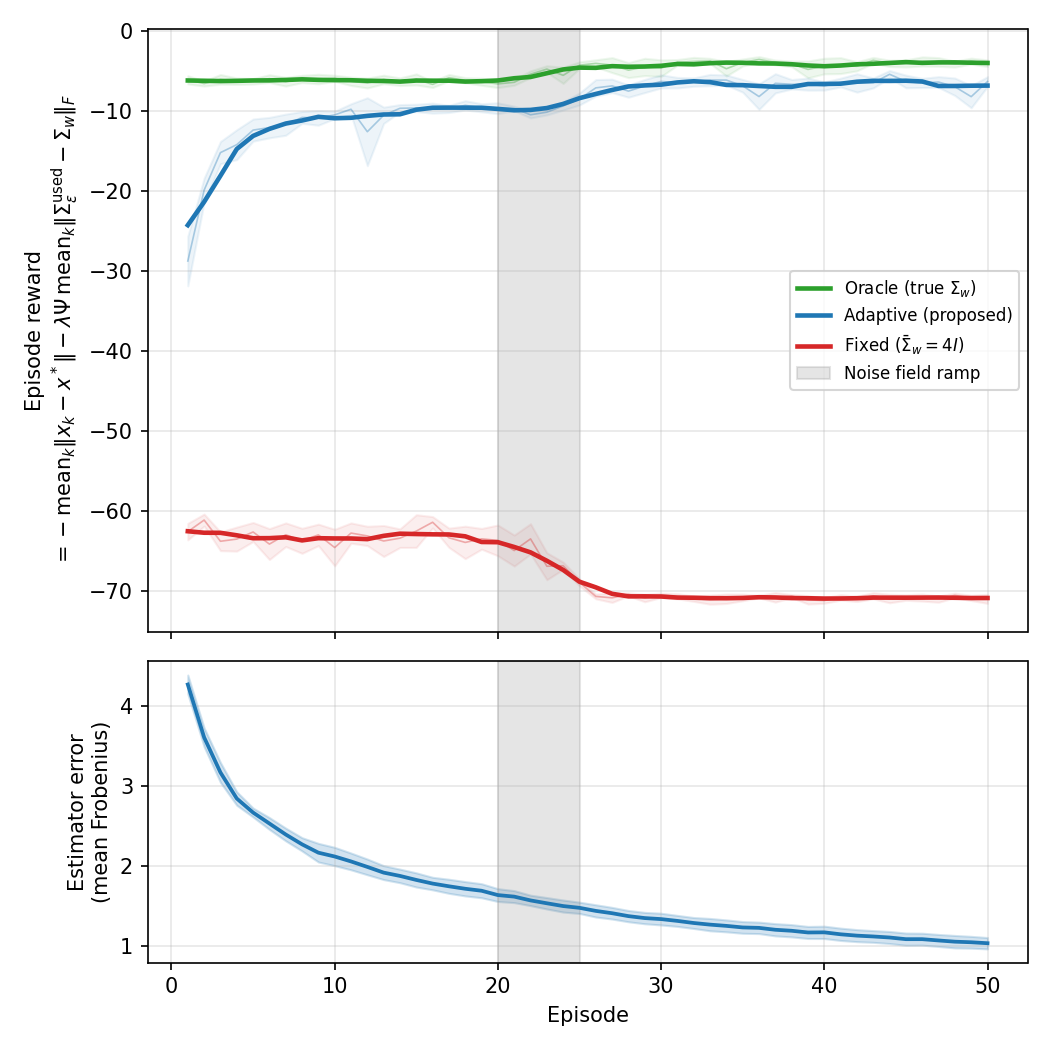}
  \caption{Scenario B: recovery after a noise field change (gray
    band), eq.~\eqref{eq:scenarioreward} with $\lambda=5$. Adaptive's
    reward continues improving through the transition; fixed degrades
    afterward since its assumption is increasingly wrong relative to
    the new field. Bottom panel: mean estimator error, continuing to
    decrease through the ramp with at most a brief slowdown.}
  \label{fig:scenarioB}
\end{figure}

\section{Discussion}
\label{sec:discussion}

\textbf{Connection to \cite{med}.}
The MED paper established (1)~the Markov property of the adaptive
RSPI framework (Proposition~3), (2)~a qualitative analysis showing
that noise mismatch distorts control strength and risk sensitivity in
coupled ways (eqs.~(12)--(13)), and (3)~empirical validation on
double integrator and bicycle model environments. The present paper
preserves all three as foundations: the RL policy for selecting
$(\hat\sigma,\lambda_\alpha)$ is replaced by estimator~\eqref{eq:update},
which targets $\Sigw(s)$ as a statistical estimand and provides a
convergence rate and a bias-variance decomposition, while the
$\sigma$-$R$ constraint~\eqref{eq:sigR} is enforced at every step.

\textbf{Kernel choice and its role.}
The stationary-measure kernel~\eqref{eq:kernel} is not merely a
convenient choice: via symmetric edge weights $a_{ss'}=a_{s's}$,
giving $\kappa(s,s')=a_{ss'}/p_s$, it is the minimal condition that
makes the diffusion dissipative in the $p_s$-weighted norm and thereby
closes the Lyapunov argument of Theorem~\ref{thm:convergence}. A
uniform kernel $\kappa(s,s')=1/|\Ncal(s)|$ satisfies dissipativity
only when the stationary distribution is uniform, which is rarely the
case in goal-directed navigation tasks. In practice, $p_s$ can be
estimated from a short warm-up phase before the
estimator~\eqref{eq:update} is activated, or replaced by any
consistent estimator without changing the asymptotic analysis.

\textbf{Distinction from two-timescale SA.}
Borkar's two-timescale theorem~\cite{borkar} would provide an
alternative proof path if $\beta$ were allowed to vanish. The
advantage of our fixed-$\beta$ approach is that diffusion remains
active throughout operation, propagating information to rarely
visited cells even asymptotically, at the cost of a nonzero bias
floor $\beta d_\kappa L_\sigma r_{\Ncal}/p_{\min}$ that is
characterizable and controlled by the payoff
condition~\eqref{eq:payoffcondition}.

\textbf{Series context.}
The linear companion paper~\cite{p1lti} established the LTI special
case, and the nonlinear companion paper~\cite{p2nonlinear} extended
it to nonlinear systems, with Corollary~3 of that paper providing the
formal interface used here. Together the three papers form a complete
theoretical foundation for MPPI as a certifiably stable
sampling-based controller across LTI, nonlinear, and noise-adaptive
settings.

\textbf{Calibration accuracy versus task performance.}
Remark~\ref{rem:taskreward} states the gap formally: the payoff
theorem (Theorem~\ref{thm:payoff}) guarantees a tighter
\emph{stability bound} as calibration error
$\normF{\hSigk{s}-\Sigw(s)}$ decreases, but does not guarantee
improved \emph{task reward}, because the coupling
constraint~\eqref{eq:sigR} forces $\Sigma_\epsilon$ to act both as a
noise model and as the proposal covariance driving each MPPI
iteration's optimization step. This gap does not shrink with more
samples: in a pure target-tracking task (no obstacles) with true
noise $\sigma=0.1$, the reward-maximizing covariance remained near
$\sigma\approx1.5$ across $M\in\{32,128,512,2048\}$ rollouts, ruling
out a finite-sample explanation and confirming that the discrepancy
is structural to the constraint itself rather than a Monte Carlo
artifact. In a separate obstacle-dense test, the true covariance
($\sigma=0.1$) produced a $15/15$ collision rate over $15$ trials,
while $\sigma\approx0.5$--$0.8$ (five to eight times larger) reduced
collisions to $13/15$. Both findings point to the same mechanism:
$\Sigma_\epsilon$ also sets the exploration radius of MPPI's sampled
trajectories, and a larger radius can improve optimization progress
or obstacle avoidance independently of whether it matches $\Sigw$.
This is consistent with established practice in path-integral
control, where injecting sampling noise larger than the true process
noise is known to improve exploration~\cite{williams}, and with the
risk-sensitive framework of~\cite{med}, where the noise strength and
risk-sensitivity parameters are tuned for behavior rather than
statistical accuracy. Figure~\ref{fig:scenarioA} illustrates the
consequence directly: under a reward combining tracking performance
with a calibration regularizer, the oracle controller (zero
calibration error by construction) does not attain the best possible
tracking performance alone, yet still attains the best combined
reward once the regularization weight is set high enough to penalize
miscalibration appropriately. The bicycle-model result in
Table~\ref{tab:terminal} -- where the adaptive controller's raw task
performance exceeded even the oracle's -- is a further, independent
illustration of the same gap. Reconciling provably accurate noise
estimation with reward-optimal exploration under
constraint~\eqref{eq:sigR} is outside the scope of the stability
guarantee proved here and is a natural direction for future work, for
instance by learning a \emph{risk-adjusted} covariance that combines
the statistical estimate $\hSigk{s}$ with an exploration bonus
calibrated to local obstacle density or cost curvature.

\textbf{Limitations.}
Assumption~\ref{ass:geometric} requires $p_{\min}>0$. For
highly goal-directed trajectories in large spaces some cells
may be rarely visited; the $d_\kappa/p_{\min}$ factor
in~\eqref{eq:bias} quantifies this degradation. The
kernel~\eqref{eq:kernel} requires $p_s$ to be known or estimated;
estimation error in $p_s$ introduces a secondary bias not analyzed
here. Future work will address directed exploration policies that
actively probe high-uncertainty cells.

% ------------------------------------------------------------------
\section{Conclusion}
\label{sec:conclusion}

We have presented an online noise covariance estimator for MPPI
with stationary-measure-weighted spatial diffusion and proved four
results: (1)~the diffusion operator with kernel~\eqref{eq:kernel}
is dissipative in the $p_s$-weighted norm
(Lemma~\ref{lem:dissipative}), which closes the Lyapunov decrease
argument without two-timescale SA; (2)~the estimator converges to
a smoothed fixed point with SA error $O(1/\sqrt{k})$, smoothing
bias $O(\beta d_\kappa L_\sigma r_{\Ncal}/p_{\min})$, and a
finite-horizon accumulated drift error $O(\epsilon_v T)$
(Theorem~\ref{thm:convergence});
(3)~the estimation error enters the companion nonlinear stability
bound~\cite{p2nonlinear} additively, via the adaptive
bound~\eqref{eq:adaptivebound}, through a computable penalty
$\psi^{(k)}$ that decays to a computable bias floor
(Proposition~\ref{prop:adaptive}); (4)~the adaptive system achieves
a strictly tighter stability bound than any fixed mismatched
controller after a crossover time $k^*$ given explicitly
in~\eqref{eq:crossover} (Theorem~\ref{thm:payoff}).

Three design rules follow. (R1)~Choose kernel~\eqref{eq:kernel}
using estimated or measured $p_s$; uniform weights are only
correct when the state space is visited uniformly.
(R2)~Verify condition~\eqref{eq:payoffcondition}: the smoothing
bias $\beta d_\kappa L_\sigma r_{\Ncal}/p_{\min}$ plus the
finite-horizon drift allowance $C_v\epsilon_v T$ must be smaller
than the ignorance error $\max_s\normF{\bar\Sigma_w-\Sigw(s)}$. As
Section~\ref{sec:simulation} illustrates, this is a real
constraint, not a formality: at $\beta=0.05$ the double
integrator's bias exceeds its tolerance by roughly $19\times$,
despite its noise field appearing visually smoother than the
bicycle model's narrow-passage discontinuity -- $L_\sigma$ must be
evaluated from Assumption~\ref{ass:smooth}'s literal, adjacent-cell
definition, not a global range estimate, or the wrong environment
will be flagged as the fragile one. (R3)~Use
Corollary~\ref{cor:samplesize} to verify $M\geq M^*_{\mathrm{adapt}}$
at deployment.

% ------------------------------------------------------------------
\bibliographystyle{IEEEtran}
\bibliography{references}

@misc{p1lti,
  author       = {Yoon, Hyung-Jin and Kim, Hunmin},
  title        = {Finite-Sample Closed-Loop Stability of Model Predictive Path Integral Control for Linear Time-Invariant Systems},
  year         = {2026},
  howpublished = {arXiv preprint arXiv:2607.04006},
  url          = {https://arxiv.org/abs/2607.04006},
  doi          = {10.48550/arXiv.2607.04006}
}

@misc{p2nonlinear,
  author       = {Yoon, Hyung-Jin and Kim, Hunmin},
  title        = {Stochastic Stability of Nonlinear MPPI via Contraction Theory and Control Lyapunov Functions},
  year         = {2026},
  howpublished = {arXiv preprint arXiv:2607.06945},
  url          = {https://arxiv.org/abs/2607.06945},
  doi          = {10.48550/arXiv.2607.06945}
}

@misc{mppi_noise_covariance_code,
  author       = {Yoon, Hyung-Jin and Kim, Hunmin},
  title        = {Simulation Code for Adaptive MPPI with Online Disturbance Covariance Estimation},
  year         = {2026},
  howpublished = {GitHub repository},
  url          = {https://github.com/LCAS-Lab/mppi-noise-covariance-estimation},
  note         = {Accessed: 2026-07-09}
}

@book{kushneryin,
  author    = {Kushner, Harold J. and Yin, G. George},
  title     = {Stochastic Approximation and Recursive Algorithms and Applications},
  edition   = {2nd},
  publisher = {Springer-Verlag},
  year      = {2003}
}

@book{borkar,
  author    = {Borkar, Vivek S.},
  title     = {Stochastic Approximation: A Dynamical Systems Viewpoint},
  publisher = {Cambridge University Press},
  address   = {Cambridge, UK},
  year      = {2008},
  isbn      = {9780521515924}
}

@article{williams,
  author    = {G. Williams and A. Aldrich and E. A. Theodorou},
  title     = {Model Predictive Path Integral Control: From Theory
               to Parallel Computation},
  journal   = {Journal of Guidance, Control, and Dynamics},
  volume    = {40},
  number    = {2},
  pages     = {344--357},
  year      = {2017},
  doi       = {10.2514/1.G001921}
}

@inproceedings{med,
  author    = {H. J. Yoon and C. Tao and H. Kim and
               N. Hovakimyan and P. Voulgaris},
  title     = {Adaptive Risk Sensitive Path Integral for Model
               Predictive Control via Reinforcement Learning},
  booktitle = {Proceedings of the 31st Mediterranean Conference
               on Control and Automation (MED)},
  pages     = {926--931},
  year      = {2023},
  doi       = {10.1109/MED59994.2023.10185876}
}

@inproceedings{pi2,
  author    = {F. Stulp and O. Sigaud},
  title     = {Path Integral Policy Improvement with Covariance
               Matrix Adaptation},
  booktitle = {Proceedings of the 29th International Conference
               on Machine Learning (ICML)},
  year      = {2012}
}

@article{odelson,
  author    = {B. J. Odelson and M. R. Rajamani and J. B. Rawlings},
  title     = {A New Autocovariance Least-Squares Method for
               Estimating Noise Covariances},
  journal   = {Automatica},
  volume    = {42},
  number    = {2},
  pages     = {303--308},
  year      = {2006},
  doi       = {10.1016/j.automatica.2005.09.009}
}

@article{dean,
  author    = {S. Dean and H. Mania and N. Matni and
               B. Recht and S. Tu},
  title     = {On the Sample Complexity of the Linear Quadratic
               Regulator},
  journal   = {Foundations of Computational Mathematics},
  volume    = {20},
  pages     = {633--679},
  year      = {2020},
  doi       = {10.1007/s10208-019-09426-y}
}

\end{document}